\documentclass[journal]{IEEEtran}

\usepackage{cite}
\usepackage{amsmath,amssymb,amsfonts}
\usepackage{algorithmic}
\usepackage{graphicx}
\usepackage{textcomp}
\usepackage{xcolor}
\usepackage{graphicx}
\usepackage{comment}
\usepackage{enumitem}
\usepackage{balance}
\usepackage{booktabs}
\usepackage{float}
\usepackage{amsthm}
\usepackage{hyperref}
\usepackage{subcaption}
\usepackage{mathrsfs}
\usepackage{breqn}

\newtheorem{definition}{Definition}[section]
\newtheorem{theorem}{Theorem}
\newtheorem{exmp}{Example}[section]
\newtheorem{lemma}[theorem]{Lemma}

\ifCLASSINFOpdf
\else
\fi
\usepackage{amsmath,amssymb}
\begin{document}
%
\title{On Random Tree Structures, Their Entropy, and Compression}
%
%
%

\author{Amirmohammad~Frazaneh,~\IEEEmembership{Student~Member,~IEEE,}
        Mihai-Alin~Badiu,~\IEEEmembership{Member,~IEEE}
        and~Justin~P.~Coon,~\IEEEmembership{Senior~Member,~IEEE}
\thanks{The authors are with the Department of Engineering Science, University of Oxford, Parks Road, Oxford, OX1 3PJ, UK, (email: amirmohammad.farzaneh@eng.ox.ac.uk; mihai.badiu@eng.ox.ac.uk; justin.coon@eng.ox.ac.uk).}
}

\maketitle

\begin{abstract}
Measuring the complexity of tree structures can be beneficial in areas that use tree data structures for storage, communication, and processing purposes. This complexity can then be used to compress tree data structures to their information-theoretic limit. Additionally, the lack of models for random generation of trees is very much felt in mathematical modeling of trees and graphs. In this paper, a number of existing tree generation models such as simply generated trees are discussed, and their information content is analysed by means of information theory and Shannon's entropy. Subsequently, a new model for generating trees based on practical appearances of trees is introduced, and an upper bound for its entropy is calculated. This model is based on selecting a random tree from possible spanning trees of graphs, which is what happens often in practice. Moving on to tree compression, we find approaches to universal tree compression of the discussed models. These approaches first transform a tree into a sequence of symbols, and then apply a dictionary-based compression method. Conditions for the universality of these method are then studied and analysed.
\end{abstract}


%
\IEEEpeerreviewmaketitle

\section{Introduction}
Trees are widely used in different areas of science. One of their biggest application area is the field of network science to model different structures, patterns, and behaviours. Some networks are formed specifically as a tree, such as the design used in the ZigBee specification \cite{ergen2004zigbee}. Additionally, trees are often encountered in practice as subsets of complex networks. An application of this is routing tables in networks, which are essentially a tree structure \cite{draves1999constructing, farzaneh2022treeexplorer}. Other notable fields in which trees are used for modelling data include phylogenetic trees \cite{felsenstein2004inferring}, parse trees in Natural Language Processing \cite{marcinkiewicz1994building}, and Barnes-Hut trees in astrophysics \cite{barnes1986hierarchical}. 

There are numerous studies that focus on the extraction of information content, often called the complexity, of graphical data structures. This is mainly because of the complex nature of graphical data structures, especially as their number of nodes increases. This exponential growth in complexity demands a formal way of quantifying the amount of information content in graphical data structures. This knowledge can later on be used for any application that involves the storage, transmission, or processing of these data structures. When it comes to data transmission, Shannon's entropy \cite{shannon1948mathematical} has been the standard metric for complexity ever since its introduction back in 1948. Some notable studies on quantifying the complexity of graphical data structures in terms of Shannon's entropy include the calculation of the entropy of Erdős–Rényi structures \cite{choi2012compression}, and random geometric graphs \cite{badiu2022structural}. Compared to the numerous studies on this area for graphs, measuring the complexity of tree structures has barely been studied before. Some tree models for which entropy has been studied before include random binary trees \cite{kieffer2009structural}, and specific cases of plane trees \cite{golebiewski2017entropy}. This shows that most of the studied models either lack generality, or lack the ability to simulate trees observed in real networks. Even though trees can be seen as a subset of graphs, their unique characteristics and features can be utilized to extract more accurate bounds and results, which can prove to be more useful. Consequently, one of the main approaches of this paper is to study the information content and complexity of tree data structures using Shannon's entropy.

Despite their vast applications, there are very few models for random generation of trees. There are numerous models for creating random graphs, such as the Erdős–Rényi model \cite{gilbert1959random}, Barabasi-Albert model \cite{albert2002statistical}, Stochastic Block Model \cite{holland1983stochastic}, and Random Geometric Graphs \cite{penrose2003random}. However, this variety can not be seen in random models for trees. The existing models for trees are very limited, and most of them rely solely on a uniform distribution among the possible trees. For example, a random generation of a Prüfer sequence \cite{prufer1918neuer} can result in a random tree. Other models focus only on specific type of tree, such as binary trees \cite{makinen1999generating}. One of the most detailed studies on random trees can be found in \cite{drmota2009random}, where several random tree models are introduced an analysed. The tress created using most of the models introduced in \cite{drmota2009random}, such as Polya trees \cite{mauldin1992polya} and Galton-Watson trees \cite{watson1875probability} can grow indefinitely, which means that they are only useful in limited application scenarios. For instance, recursive trees are useful in the analysis of the spread of epidemics \cite{moon1974distance}, family trees of preserved copies of ancient manuscripts \cite{najock1982number}, and chain letters or pyramid games \cite{gastwirth1984two}, all of which have the possibility of growing indefinitely. One of the most simple, yet powerful, models that has been introduced so far for random tree generation is called Simply Generated Trees (or unconditioned Galton-Watson trees). This model is simple, but still effective to capture the dynamics of trees. It has been shown that Simply Generated Trees can act as a generation model for many different types of random trees that are seen in practice \cite{janson2012simply}. One example of this is modelling branching processes \cite{devroye1998branching}. In this paper we start by focusing on this model, and extract its information theoretic content in terms of Shannon's entropy. This can be beneficial in analysing the situations in which Simply Generated Trees are used to model the generation of trees. Another way of randomly generating trees, which has not been worked on in detail in the past, is the extraction of a random tree from an already existing underlying network. This is a scenario that happens a lot in practice. For instance, as stated earlier, the tree corresponding to the routing table of a node in a network is simply a spanning tree of the original network. Therefore, it will be beneficial to look at this method of developing random trees, and analyse it using information theory. In this paper, we formalize this method of creating random trees, and then find an upper bound to its entropy.

Ultimately, we focus on compressing trees to their theoretical limit, which is the entropy of the source. Tree compression can be used in all applications that involve the use of trees. For instance, applications of tree compression in syntax-directed compression of program file and pixel trees are studies in \cite{katajainen1990tree}. The storage space and bandwidth required for communicating trees grows exponentially large with the size of the tree if traditional methods such as adjacency matrices or adjacency lists are used for coding trees. Tree coding methods such as the Prüfer code, the Neville codes, and the Demo and Micikevičius code have a fixed codeword length for same trees even if they are generated from different distributions \cite{caminiti2007coding}, and therefore are not entropy-optimal. Additionally, survey methods that do study information-theoretic optimality of the introduced coding methods often do so by comparing the average codeword length with the entropy of the uniform source only \cite{katajainen1990tree}. Recently, powerful tree compression algorithms such as tree compression with top trees \cite{bille2015tree} tree structure compression using RePair \cite{lohrey2013xml} have been introduced. However, the performance of these algorithms is not analysed using information thoery, which always leaves the question of whether the structures could be compressed more efficiently. For these reason, we seek compression algorithm for trees that are both easy to implement and are shown to have a near optimal performance with respect to information-theoretic measures. The family of dictionary-based compression methods are among the most preferred lossless compression algorithms for different sources and data types \cite{hosseini2012survey}. They are widely used in applications such as text compression \cite{bell1989modeling} and image compression \cite{rahman2019lossless}. Because of the advantages of dictionary-based compression methods, we want to be able to use them to optimally compress tree structures as well. However, the issue is that a tree structure is not made out of a consecutive series of symbols. For this reason, we first consider a family of transformations on trees, which we call tree traversals. We give examples of possible ways that tree traversals can be done, and then move on to applying dictionary-based compression algorithms to sequences generated using tree traversals. We then discuss under what conditions this method of tree compression guarantees optimality.

The paper starts with the study of the complexity of different tree sources. We start with the uniform source, and then move on to the entropy of Simply Generated Trees. Afterwards, a model for generating random trees using an underlying random graph is introduced, which we call the Spanning Tree Model. An upper bound to the entropy of a special case of the spanning tree model is also calculated. We then move on to the subject of tree compression, by introducing tree traversal and examples of it such as Pit-Climbing and Tunnel-Digging. Subsequently, we combine tree traversals and dictionary-based compression algorithms to Simply Generated Trees and the spanning tree model, and show that this combination can optimally compress trees generated from these sources. The paper ends with a conclusion, and potential future directions of research are mentioned.

\section{Entropy of Tree Structures}

In this section, we focus on quantifying the entropy of tree structures. If we look at trees as a random variable from a pool of possible trees, we need a probability distribution on these trees to be able to calculate Shannon's entropy. This necessitates the existence of a random model for creating the tree structures at hand. Having a random model would entail having a probability distribution on the trees, and model the dynamics of creating trees in real-life scenarios. Unfortunately, the issue is that there does not exist an adequate variety of random models for creating tree structures. Existing models are very limited in terms of scenarios they can simulate. For instance, there is no parameter to set the number of nodes in Galton-Watson trees, and the number of nodes in the tree can be anywhere from one to infinity. We will start this section by introducing some concepts and terminologies that we use throughout the paper. We then continue with a very simple method for creating random trees, which is to choose a tree uniformly among all possible trees with the same number of nodes. Additionally, One of the few existing models for creating random trees, called Simply Generated Trees, is studied. We will then introduce another model for creating a random tree. This model is based on creating the underlying graph first, and then choosing one of its spanning trees as the output of the random generator. We call this model the spanning tree model, and then study its entropy.

\subsection{Terminologies}
We define a random tree source using the following parameters.
\begin{itemize}
    \item $T$: set of all possible trees that can be generated by the source. The size of this set can be finite or infinite.
    \item $p_T(t)$: A probability distribution on the trees in $T$. We may only use the notation $p(t)$ if it is clear which tree source we are talking about. Generally, this distribution can be time-dependent or time-independent.
\end{itemize}

In this paper, the term entropy always refers to Shannon's entropy, and is calculated in base two (bits).

\subsection{Entropy of the uniform source}
\label{entropy_uniform}

In this section, we calculate the entropy of a uniform source for various types of tree structures. By a uniform source we mean that having $T$, the probability of observing any $t\in T$ is simply $1/|T|$. This way, the entropy of the source will simply be $\log_2 |T|$ bits.

We first start with unlabeled unordered rooted trees. The sequence of the number of unlabeled unordered rooted trees with $n$ nodes is listed on the on-line encyclopedia of integer sequences \cite[A000081]{oeis}. Table \ref{sequence} shows the possible trees of this kind that can be built with up to three nodes. It is known that the asymptotic limit of this sequence is $cd^nn^{-3/2}$ \cite{polya2012combinatorial}, where $c$ and $d$ are constants that can be found at \cite[A187770]{oeis} and \cite[A051491]{oeis}, respectively. Consequently, the asymptotic uniform entropy for this model, $H_T$, can be calculated using the following equation. This equation shows that the growth rate of the entropy of a uniformly distributed unlabeled unordered rooted tree source is asymptotically linear.

\begin{align}
\begin{split}
\label{asymp_entropy}
    H_T&\sim n\log_2d-1.5\log_2n+\log_2c\\
    &\approx 1.5635n-1.5\log_2n-1.1846
\end{split}
\end{align}

To study the same for unlabeled ordered rooted trees, we refer to the fact that their number matches the sequence of Catalan numbers \cite[Ch.~8]{koshy2008catalan}. In other words, the $n$th Catalan number equals the number of possible ordered rooted trees that can be built on $n$ nodes. The following equality is well-known about Catalan numbers \cite{sedgewick1996introduction}.
\begin{equation}
\label{Catalan}
\begin{split}
    C_n &= \frac{1}{n+1}{2n \choose n}\\
    &=\frac{4^n}{\sqrt{\pi n^3}}\left(1+O\left(\frac{1}{n}\right)\right)
\end{split}
\end{equation}
Based on (\ref{Catalan}), the uniform entropy of the source can be calculated using the following equation.
\begin{equation}
\label{Catalan2}
\begin{split}
H_T &= \log_2 C_n\\
& = \log_2 \frac{1}{n+1}{2n \choose n}\\
 &=\log_2 \frac{4^n}{\sqrt{\pi n^3}}\left(1+O\left(\frac{1}{n}\right)\right)
\end{split}
\end{equation}

Additionally, the asymptotic behaviour of this entropy can also be analysed using (\ref{Catalan2}), which would provide us with the following result.
\begin{equation}
    H_T\sim 2n
\end{equation}
An interesting thing to note is that \cite{kieffer2009structural} achieves the same result for the uniform entropy, but for rooted full binary trees with $2n-1$ nodes. This suggests the existence of a one-to-one mapping between ordered rooted trees with $n$ nodes, and rooted full binary trees with $2n-1$ nodes. It is already known that a transformation called child-sibling representation maps ordered rooted trees to binary trees \cite{fredman1986pairing}. Even though this transformation does not create full binary trees and is not a bijection, we can create a bijection between ordered rooted trees and rooted full binary trees using a similar idea. We call this transformation the double-node transformation, as every node in the original rooted tree, except for the root, will be mapped into two nodes in the binary tree. This is why the $n$ nodes of the original tree will be mapped into $2n-1$ nodes in the binary tree. This transformation is a bijection as every rooted ordered tree can be mapped into a unique rooted full binary tree and vice-versa. The steps of this transformation are explained below, and Fig. \ref{double-node} shows an example of this transformation.

\begin{description}
    \item[Double-node transformation]   
    \item[Input:] a rooted ordered tree with $n$ nodes
    \item[Output:] a rooted full binary tree with $2n-1$ nodes
    \begin{enumerate}
        \item The root of the input is mapped to the root of the output. As the root does not have any siblings, it is only considered a child node.
        \item The nodes of the input are traversed using BFS. 
        \item Every observed node will be transformed into two nodes, which will both be attached to the same node in the binary tree. We call the left one the child node, and the right one the sibling node.
        \item If the observed node is the first child of its parent, its corresponding child and sibling nodes will be attached to the child node of its parent. If not, its child and sibling nodes will be attached to the sibling node of its closest sibling to its left.
        \item The mapping continues until all the nodes in the input tree have been traversed.
    \end{enumerate}
\end{description}
\begin{figure}
    \centering
    \includegraphics[width = \columnwidth]{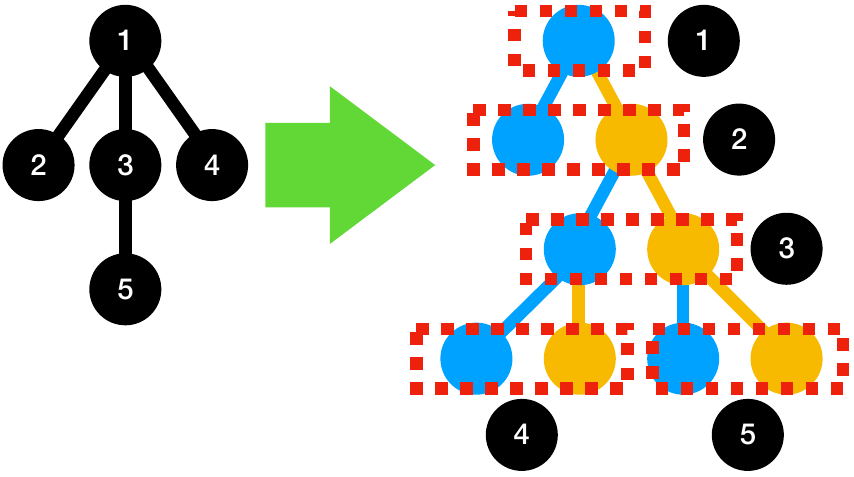}
    \caption{Example of Double-node transformation}
    \label{double-node}
\end{figure}
It can be seen that the double node transformation provides a bijection between rooted ordered trees with $n$ and rooted full binary trees with $2n-1$ nodes. Therefore, the uniform entropy results for these two families of trees are the same.

Finally we study labeled rooted and unrooted trees. The number of labeled unrooted trees is shown to be $n^{n-2}$ \cite[p.~26]{cayley_2009}, which implies that the uniform entropy can be calculated using the following equation. 
\begin{equation}
    H_T = (n-2)\log_2n
\end{equation}
Additionally, for any given labeled unrooted tree with $n$ nodes, we can pick any of its $n$ nodes as the root. Therefore, the number of labeled rooted trees is $n^{n-1}$, and their uniform entropy can be calculated using the following equation.
\begin{equation}
    H_T = (n-1)\log_2n
\end{equation}

\begin{table}
\centering
\caption{Possible unlabeled unordered rooted trees of up to three nodes}
\label{sequence}
\begin{tabular}{ |c| c| c| }
\hline
 $\boldsymbol{n}$ & \textbf{Possible trees} & \textbf{Count}\\ 
 \hline
 1 & \begin{minipage}{.4\columnwidth}
      \includegraphics[width=\linewidth]{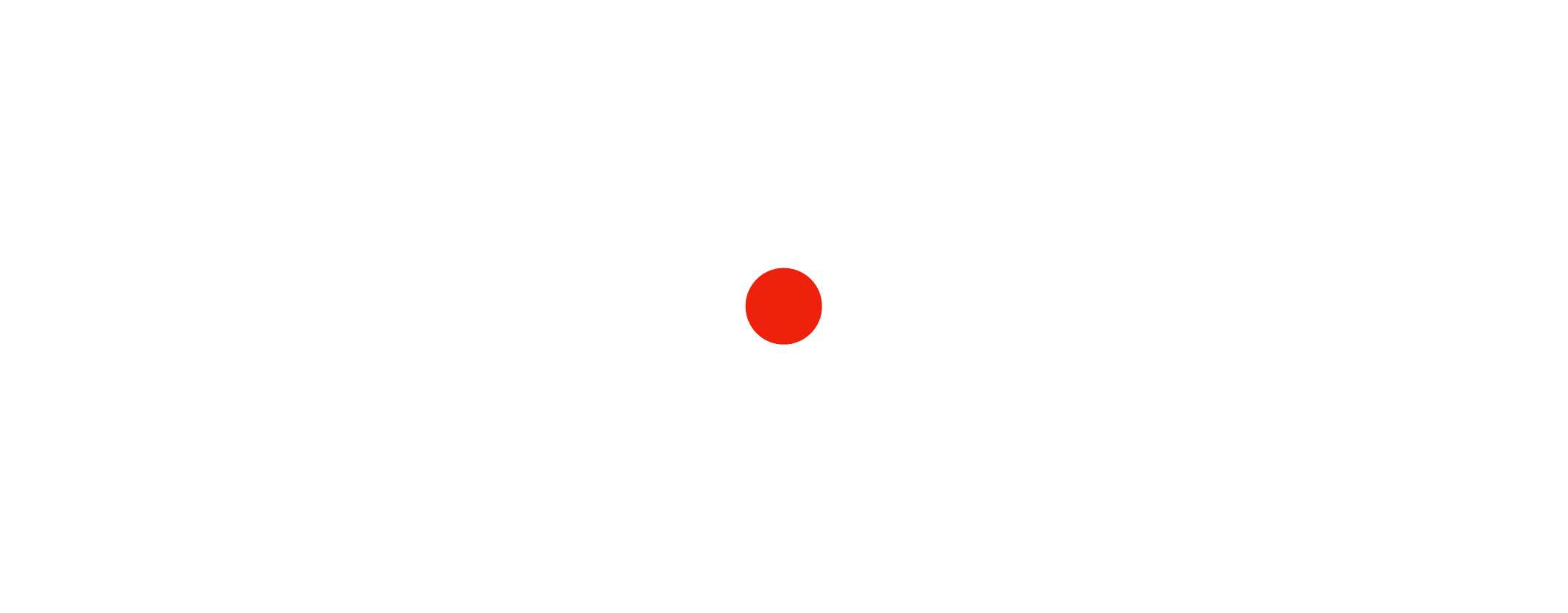}
    \end{minipage}
& 1\\
\hline
2 &  \begin{minipage}{.4\columnwidth}
      \includegraphics[width=\linewidth]{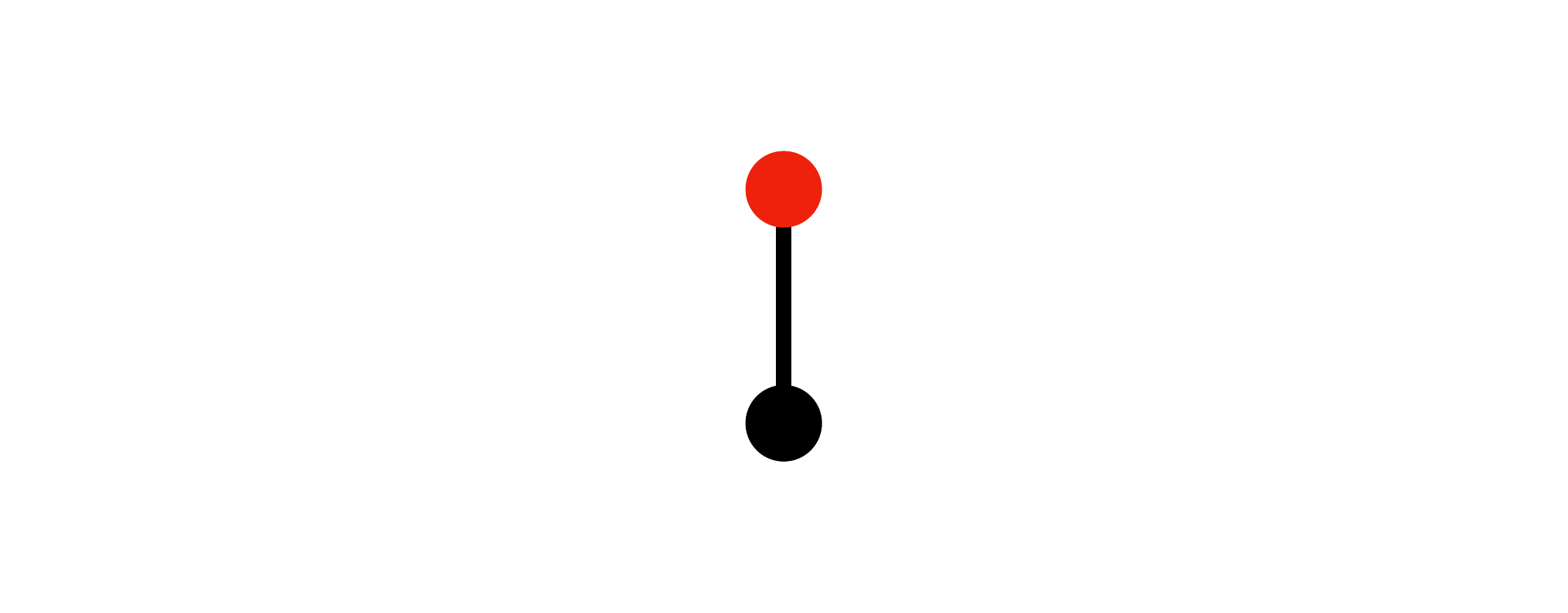}
    \end{minipage} & 1\\
    \hline
3 &  \begin{minipage}{.4\columnwidth}
      \includegraphics[width=\linewidth]{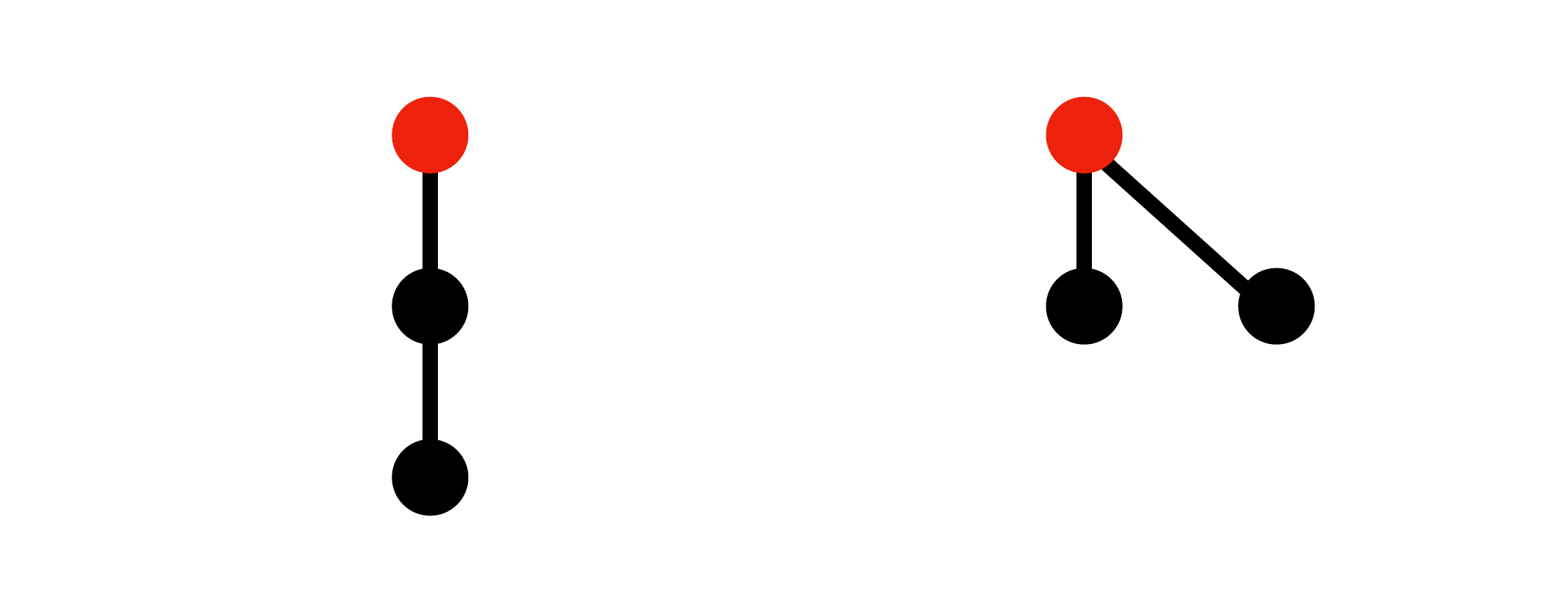}
    \end{minipage} & 2\\
    \hline
4 &  \begin{minipage}{.4\columnwidth}
      \includegraphics[width=\linewidth]{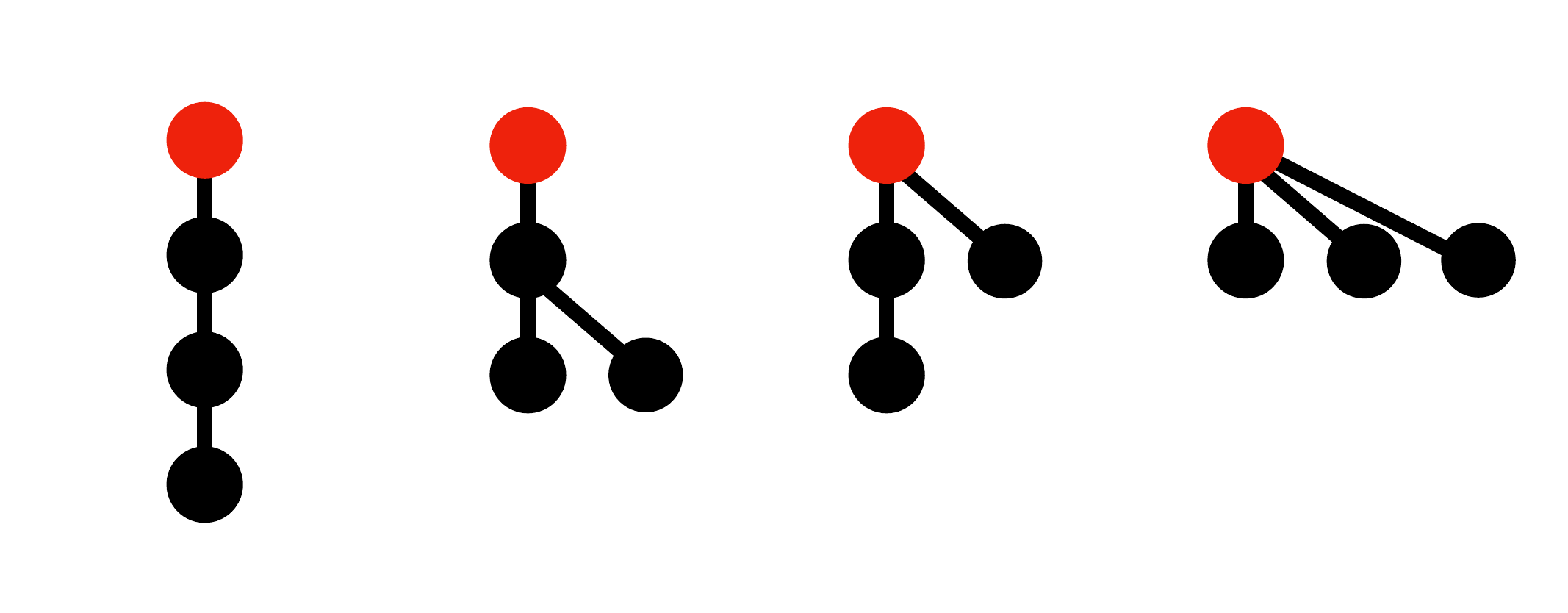}
    \end{minipage} & 2\\
    \hline    
\end{tabular}
\end{table}

\subsection{Entropy of Simply Generated Trees}
\label{Simply Generated Trees}

Simply Generated Trees (SGT) are a popular family of random trees. They were first introduced in \cite{meir1978altitude}, and since then have been used to model random trees. The reason for the popularity of this model lies in its simplicity, and the fact that it is powerful enough to model many scenarios. They are also used to generate more complex tree generation models. In this section, we will calculate the entropy of this family of random tree models.

Simply Generated Trees are generated based on a probability distribution on the number of children of each node. They are rooted and ordered trees. To build an SGT, we need a distribution on the set of whole numbers $\{0,1,2,\ldots\}$. We call this distribution the children distribution of the SGT, and show it with $p_C(c)$. The children distribution is essentially a distribution on the number of children that each node can have, independently from others. The only condition that we impose on the children distribution is for $p_C(0)$ to be nonzero. For example, we can use a geometric distribution, binomial distribution, or generally any discrete distribution on whole numbers that satisfies $p_C(0)\neq 0$. After this distribution is chosen, we are ready to generate the random tree. The following steps show how an SGT is created using its underlying children distribution.

\begin{enumerate}
    \item Create the root of the tree. This will be level 0.
    \item Create the children of the root, based on a number acquired from the children distribution.
    \item For all $i>0$, go through the nodes on level $i$, and choose the number of children for each of them, based on the children distribution.
    \item The algorithm is terminated once the number of children for all the nodes is chosen, and there is no more node to explore (Note that once a node is decided to have zero children, the branch corresponding to that node will be terminated, and the node will become a leaf node).
    
\end{enumerate}

Note that the resulting tree will have ordered branches, as the number of children for each node is chosen sequentially. Generating a tree this way will result in a tree whose probability is equal to the product of the probability of the number of children of all of its nodes. The following example illustrates SGTs that can be made using a specific children distribution.

\begin{exmp}
Assume that we have a children distribution for which $p_C(0)=p$, and $p_C(1)=1-p$. In the trees generated using this distribution, nodes can therefore have no child, or have only one child. Table \ref{SGTexmp} illustrates the possible trees that can be made using this model, alongside their probabilities.
\begin{table}[h!]
\caption{Example of Simply Generated Trees and their probabilities}
\label{SGTexmp}
     \begin{center}
     \begin{tabular}{ | c | c | }
     \hline
      Graph & Probability \\ \hline
     \begin{minipage}{.1\textwidth}
      \includegraphics[width=\linewidth]{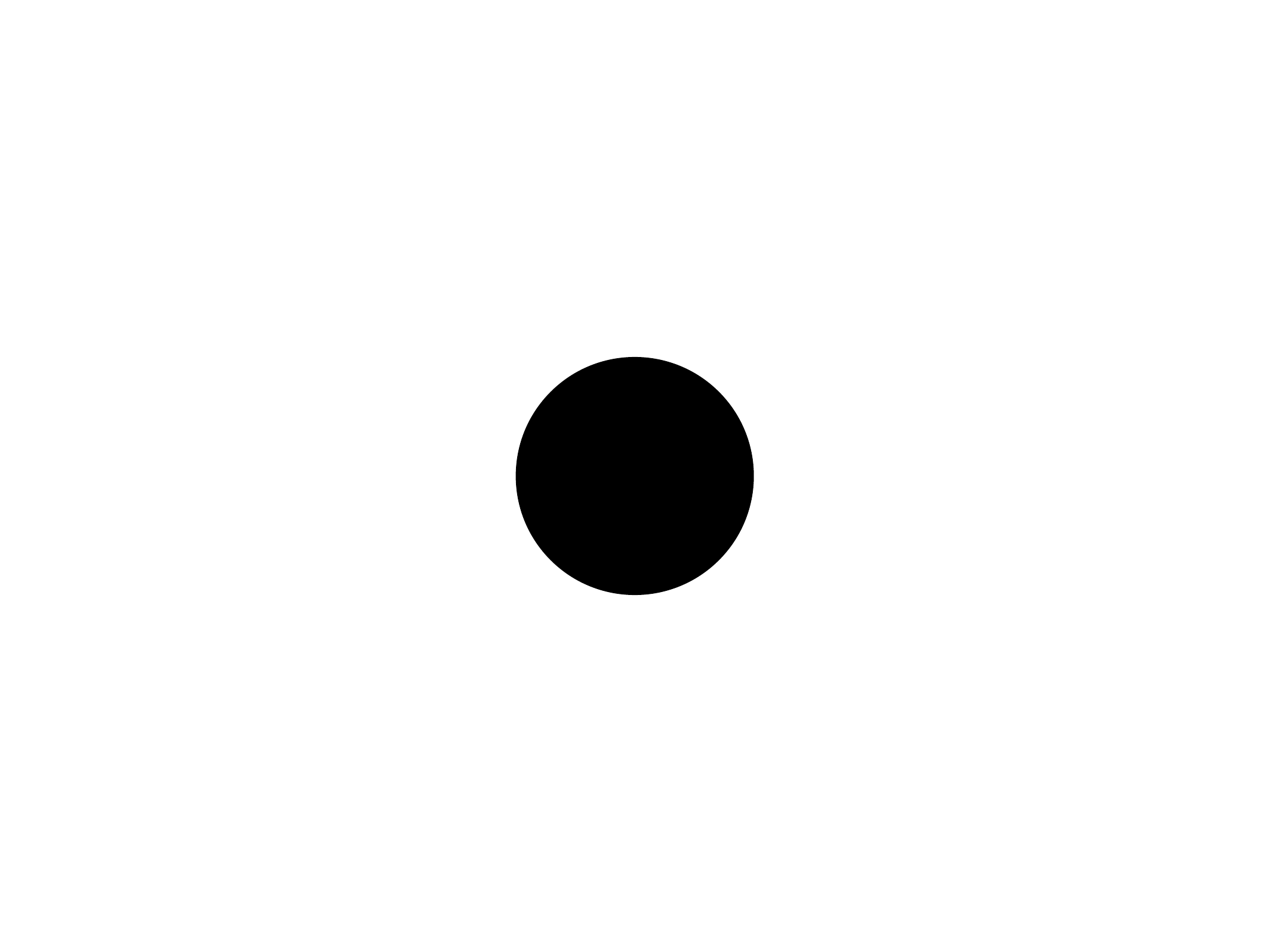}
        \end{minipage}
     & $p$\\
     \hline
     \begin{minipage}{.1\textwidth}
      \includegraphics[width=\linewidth]{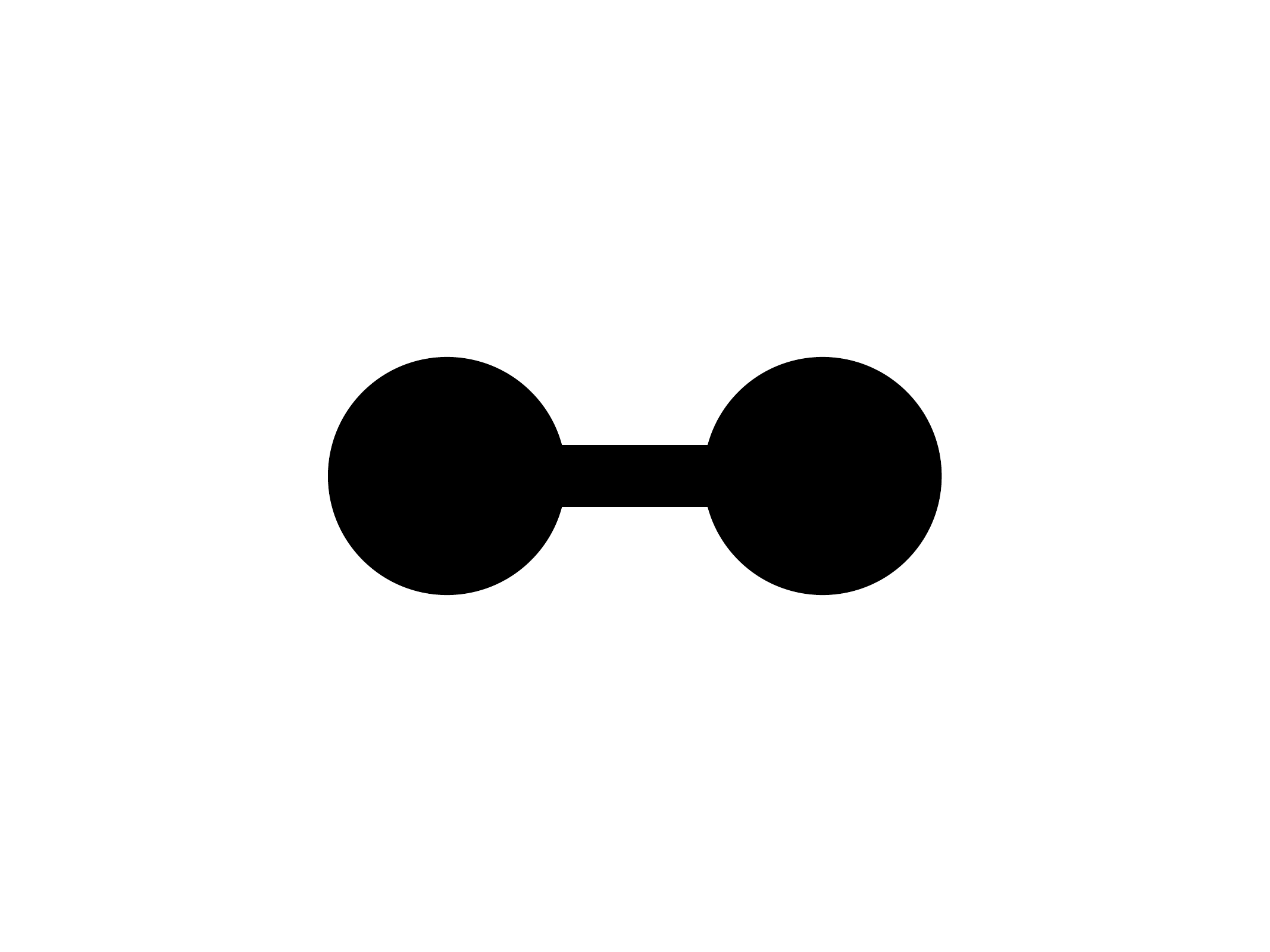}
        \end{minipage}
     & $p(1-p)$\\
     \hline
     \begin{minipage}{.1\textwidth}
      \includegraphics[width=\linewidth]{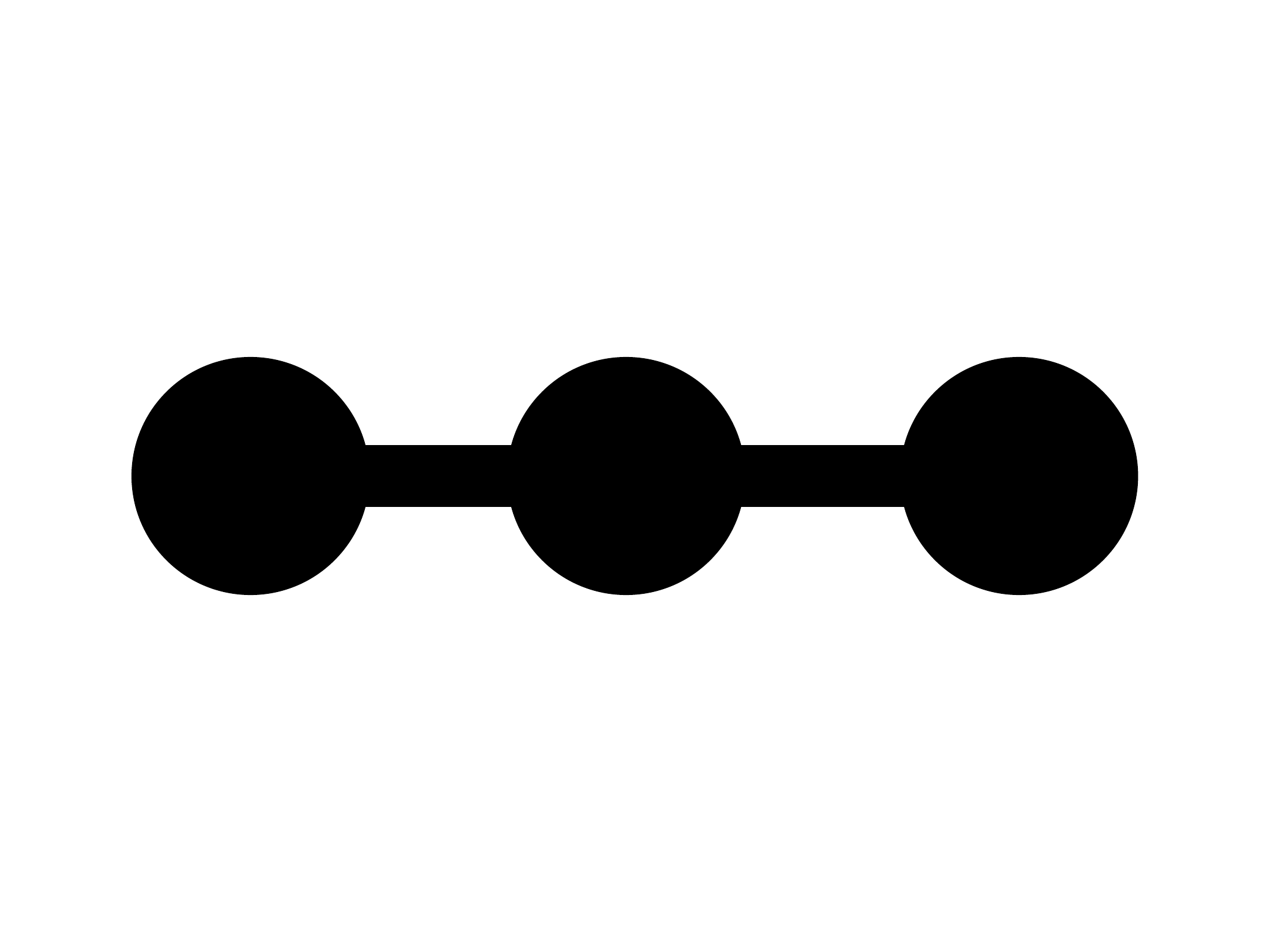}
        \end{minipage}
     & $p(1-p)^2$\\
     \hline
     Line with $n$ nodes & $p(1-p)^{n-1}$\\
     \hline
     \end{tabular}
     \end{center}
\end{table}
\end{exmp}

We now want to use information theory to measure the information content of SGTs. We will first start by studying the possible number of nodes in an SGT. It can easily be seen that in case $p_C(0)\neq 1$, the number of nodes in an SGT can theoretically be unlimited. This is because as soon as the probability of having more than zero children is nonzero, the branches can keep growing infinitely. However, we can still study the average number of nodes in SGTs.

Assume that we have an SGT model, with a children distribution $p_C(c)$, where $p_C(0)\neq 1$. Additionally, assume $T$ to be the set of all possible trees that can be generated by this model. We are interested in calculating the average number of nodes of trees in $T$. Based on the average number of children, Simply Generated Trees can be categorized into subcritical, critical, and supercriticial. This categorization can by calculating the expected number of children $\mathbb{E}[c]$, which we show with $\bar{C}$, based on the children distribution. The model is in the subcritical, critical, and supercritical if and only if $\bar{C}<1$, $\bar{C}=1$, $\bar{C}>1$, respectively \cite{janson2012simply}. It will be shown that the average number of nodes for an SGT model is limited only in the subcritical phase.

Let $n$ be the random variable that shows the number of nodes in a tree, and let $n_t$ show the number of nodes in a particular tree $t$. We can use the following equation to calculate the average number of nodes in the trees created using this model.
\begin{equation}
    \label{avg_node1}
    \mathbb{E}[n] = \sum_{t\in T}n_tp(t)
\end{equation}
For each tree $t$, $n_t$ can be written as the sum of the nodes in different levels, with level 0 being the root node. If $n_{i,t}$ shows the number of nodes in the $i$th level of tree $t$, (\ref{avg_node1}) can be rewritten as
\begin{subequations}
\begin{align}
    \label{avg_node3}
    \mathbb{E}[n] &= \sum_{t\in T}p(t)\sum_i n_{i,t}\\
    \label{avg_node2}
    &=\sum_i\sum_{t\in T}p(t)n_{i,t}.
\end{align}
\end{subequations}
Based on (\ref{avg_node2}), the average number of nodes can be rewritten as the sum of the average number of nodes in each level among all possible trees. Therefore, we need to calculate the average number of nodes for each level individually. We will show the random variable for the number of nodes in level $i$ with $n_i$. Firstly, as all the trees will have only one node in level 0, we have $\mathbb{E}[n_0]=1$. Additionally, if a tree $t$ has $n_{i-1,t}$ nodes in level $i-1$, its expected number of nodes in level $i$ will simply be $n_{i-1,t}\bar{C}$. This is because the number of children of individual nodes is independent. Therefore, if we show the probability of a generated tree having $i$ nodes in level $n$ with $p_n(i)$, we can write the following recursive equation for $i>0$.
\begin{subequations}
\begin{align}
    \label{recursive_nodes1}
    \mathbb{E}[n_i]&= \sum_{j=0}^\infty p_{n-1}(j)j\bar{C}\\
    \label{recursive_nodes2}
    &= \bar{C}\sum_{j=0}^\infty p_{n-1}(j)j\\
    \label{recursive_nodes}
    &= \bar{C}\mathbb{E}[n_{i-1}]
\end{align}
\end{subequations}
Based on (\ref{recursive_nodes}) and the fact that $\mathbb{E}[n_0]=1$, we have
\begin{equation}
    \label{level_num}
    \mathbb{E}[n_i] = \bar{C}^i.
\end{equation}
Combining (\ref{avg_node2}) and (\ref{level_num}), we get the following equation for the average number of nodes in the SGT model.
\begin{equation}
    \label{En}
    \mathbb{E}[n] = \sum_{i=0}^\infty \bar{C}^i
\end{equation}
As (\ref{En}) is a geometric series, we know that it will only converge when we have $|\bar{C}|<1$. This means that the SGT model needs to be in the subcritical phase for the average number of nodes of the model to converge. From now on, we will assume the SGT models that we work with to be in the subcritical phase, so that the corresponding trees have a finite average node number.

Now, we move on to quantifying the entropy of an SGT model based on its children distribution. Assume the corresponding children distribution of the SGT model to have an entropy of $H_C$. In other words, we have
\begin{equation}
    H_C = -\sum_{i=0}^\infty p_C(i)\log_2p_C(i)
\end{equation}
We want to represent the entropy of the SGT model based on $H_C$. If we show the entropy of the SGT model with $H_T$, we can write
\begin{equation}
    \label{ent1}
    H_T = -\sum_{t\in T}p(t)\log_2p(t).
\end{equation}
To simplify (\ref{ent1}), we note that we can write $p(t)$ as the product of the probabilities of the number of children of each node in $t$. We use this to write $\log_2p(t)$ in the form of a sum. We use $n_t$ to show the number of nodes in $t$, and $c_{t,i}$ to show the number of children of node $i$ in $t$.
\begin{equation}
    \label{ent2}
    H_T = -\sum_{t\in T}p(t)\sum_{i=1}^{n_t}\log_2p(c_{t,i})
\end{equation}
To simplify (\ref{ent2}), we open it up based on the number of children. We use $n_{t,i}$ to show the number of nodes in $t$ that have $i$ children. This way we can write
\begin{equation}
    \label{ent3}
    H_T = \sum_{i=0}^{\infty} H_{T,i},
\end{equation}
where
\begin{equation}
    \label{ent4}
    H_{T,i} = -\log_2p(i)\sum_{t\in T}p(t)n_{t,i}.
\end{equation}
To calculate $H_{T,i}$, we note that $\sum_{t\in T}p(t)n_{t,i}$ essentially represents the average number of nodes that have $i$ children in a random tree. To calculate this sum, we perform the same trick as when we calculated the average number of nodes in a tree, we open up the sum on the level of nodes in the tree. It can easily be seen that the average number of nodes that have $i$ children in a specific level of the tree is simply $p(i)$ times the number of nodes on that level. Therefore, we can use (\ref{recursive_nodes1}) and write

\begin{equation}
\label{ent5}
    \sum_{t\in T}p(t)n_{t,i} = p(i)\sum_{i=0}^{\infty}\bar{C}^i = \frac{p(i)}{1-\bar{C}}.
\end{equation}

Inserting (\ref{ent5}) into (\ref{ent4}), we get the following equation.
\begin{equation}
    \label{ent6}
     H_{T,i} = \frac{-p(i)\log_2p(i)}{1-\bar{C}}
\end{equation}
Finally, to calculate $H_T$, we insert (\ref{ent6}) into (\ref{ent3}), and we get the following equation.
\begin{subequations}
\begin{align}
    \label{ent71}
     H_{T} &= -\sum_{i=0}^\infty \frac{p(i)\log_2p(i)}{1-\bar{C}}\\
    \label{ent7}
    & = \frac{H_C}{1-\bar{C}}
\end{align}
\end{subequations}
(\ref{ent7}) gives us the entropy of the SGT model. It can be seen that the entropy can easily be calculated from the entropy of the underlying children distribution. Additionally, note that in order for this entropy to be computable, the sum in (\ref{ent5}) needs to converge, and therefore the SGT needs to be in the subcritical phase. Additionally, note that this entropy can be seen as the number of nodes times $H_C$, which makes sense intuitively. This is because the number of children of each node is independent from the others, and adds an entropy of $H_C$ to the ensemble. This is a very interesting result, as it relates the entropy of the underlying children distribution to the tree model using a simple equation.

It can be seen in (\ref{ent7}) that both the entropy of the children distribution and the average number of children have an effect on the entropy of the trees. An increase in the entropy of the children distribution or the average number of children per node results in an increase in the entropy of the trees. Additionally, as $0\leq \bar{C} \leq 1$, it can be concluded that $H_T\geq H_C$, with equality holding if and only if the the only possible tree that the SGT model can create is a tree with only one node.

\subsection{Conditioned Galton-Watson Trees}

In this section, we will attempt to quantify the entropy of conditioned Galton-Watson trees. Conditioned Galton-Watson trees are simply Galton-Watson trees that are conditioned on their number of nodes. In other words, we condition the trees generated using the Galton-Watson model such that $|t|=n$. This will change the probability distribution of the original Galton-Watson trees. 

Consider the random variables $\{X_1,X_2,\ldots,X_n\}$ to be $n$ i.i.d random variables sampled from the children distribution. For this sequence of random variables to be able to represent a tree with $n$ nodes, we need to have
\[\sum_{i=1}^nX_i=n-1.\]
Therefore, the entropy that we are looking for can be formed using the following conditional entropy.
\begin{equation}
\label{conditioned1}
    H_n = H(X_1,\ldots,X_n|\sum_{i=1}^nX_i=n-1)
\end{equation}
Calculating the entropy using (\ref{conditioned1}) can prove to be challenging. Therefore, we will provide estimates to it using upper bound. We can start by using a zero-order upper bound as below.
\begin{equation}
\label{zero_order}
\begin{split}
    H(X_1,\ldots,X_n|\sum_{i=1}^nX_i=n-1)&\leq H(X_1,\ldots,X_n)\\
    & \leq \sum_{i=1}^nH(X_i)\\
    & = nH_C
\end{split}
\end{equation}
Eq. (\ref{zero_order}) provides us with a simple upper bound. If we want to increase the accuracy of the upper bound, we can move up to higher order models. The equation below describes the first order model.
\begin{equation}
\label{first_order}
\begin{split}
    &H(X_1,\ldots,X_n|\sum_{i=1}^nX_i=n-1)\\
    &\leq H(X_1|\sum_{i=1}^nX_i=n-1)+\sum_{i=1}^nH(X_i|X_1,\sum_{i=1}^nX_i=n-1)\\
    & \leq H_{C,n-1}+(n-1)\mathbb{E}_{C,n-1}\left[H_{C,n-1-i}\right],
\end{split}
\end{equation}
where $H_{C,i}$ and $\mathbb{E}_{C,i}$ show the entropy of the children distribution conditioned on the number of children being limited to $i$ and the expectation over this conditional distribution, respectively. The method used in (\ref{first_order}) can provide us with a tighter upper bound, but has more computational complexity. We can also keep increasing the order in the same manner to get more accurate bounds at the cost of more computational complexity.

\subsection{The Spanning Tree Model}

Unlike graphs, the models for generating random trees are very limited. The variety observed in random graph models can not be seen in trees. Different distributions are fit to real-life networks in order to create mathematical models that can capture the properties of these networks. For instance, the Watts-Strogatz model \cite{watts1998collective} exhibits a high clustering coefficient, which is consistent with many real-life networks. However, random models that are able to exhibit the same properties of real-life tree data structures do not exist. Because of this reason, we were motivated to define a new random tree generation model, which is also based on practice. As a result, we introduce and study the spanning tree model in this section.

In practical applications, the trees we work with are often a spanning tree of a network. An example of this can be seen in network routing. Routing tables are usually used to store the shortest paths from any node in a network to any other one. It can be shown that the routing table for a node in a network is essentially representing a rooted tree, with the root being the origin node. It can also be seen that if the network forms a connected graph, this rooted tree is a spanning tree of the underlying network. Therefore, selecting a random spanning tree of the underlying network can be a practical model for generating random trees. The combination of random graph generation with random spanning tree selectors for creating random trees is a novel approach. One of the key points of this combination is that both fields have been studied extensively and powerful tools exist in both of them. We have already discussed the variety of different models that exist for random graph generation. In addition to that, there are powerful ways of randomly selecting a spanning tree of a single graph \cite{broder1989generating, kelner2009faster, madry2014fast}. The introduced model is very flexible in the sense that any existing random graph generation model can be combined with any existing random spanning tree selector in order to fit the real-life scenario that we want to simulate. In this section, we analyse the entropy of the trees that are created using this method.

Assume that we have a random graph source $G$. The first step is to create a graph $g$, according to the distribution of $G$. $g$ will then have a number of spanning trees (which can also be zero). We then choose one of the spanning trees of $g$ as the generated tree. This can generally be done according to an arbitrary distribution, which can also be dependent on $g$. We use the term spanning tree model to refer to this method of generating a random tree. We use $H_G$ to show the entropy of the random graph generator used in the model, and $H_T$ to show the entropy of the trees generated using the spanning tree model. Fig. \ref{spanning tree} shows the steps of the spanning tree model.

\begin{figure}
    \centering
    \includegraphics[width = \columnwidth]{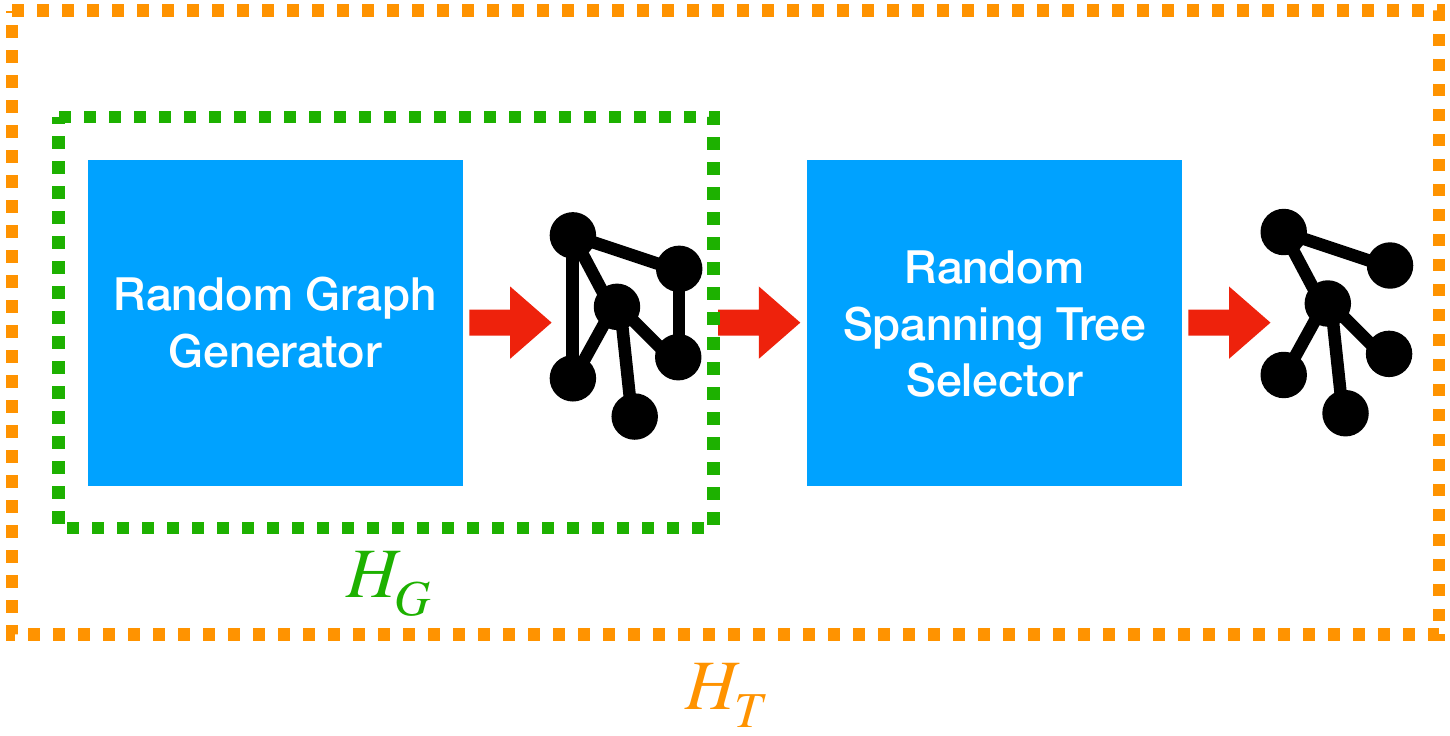}
    \caption{Steps of the Spanning Tree Model}
    \label{spanning tree}
\end{figure}

The goal in this section is to find $H_T$, assuming that $H_G$ is known or can be easily calculated. We can write the following equation to find $H_T$.

\begin{equation}
    \label{Ht}
    H_T = H_G + H(T|G)-H(G|T)
\end{equation}

If we assume that the output trees are chosen uniformly from the spanning trees of the graph, then we can write the following equation for calculating $H(T|G)$.

\begin{subequations}
\begin{align}
    H(T|G) &= \sum_{g\in G}p_G(g)H(T|G = g)\\
    \label{uniform Sg}
     &= \sum_{g\in G}p_G(g)\log_2s_g,
\end{align}
\end{subequations}
where the sum is taken among all the possible graphs that can be generated using the random graph generator, $p_G$ shows the probability distribution of the graph source, and $s_g$ shows the number of spanning trees of graph $g$. Notice that the summation needs to be taken among connected graphs of the source, as they need to have at least one spanning tree for $\log_2s_g$ to be defined.

(\ref{Ht}) and (\ref{uniform Sg}) show that to calculate the entropy of the spanning tree model, we need knowledge about the underlying distribution for the network topology, as well as the number of spanning trees that exists for each graph. Unfortunately, the number of spanning trees of a graph, often called its tree-number, does not have a closed form representation in terms of its number of nodes and edges. There are methods such as Kirchhoff's theorem \cite[Theorem.~13.1]{moore2011nature}, that provide us with a way for calculating the tree-number of graphs, but they require knowledge about the graph's adjacency information. The spanning tree entropy also depends on the model that is used to generate the underlying graph topology, which is not always known. Because of these limitations, we will consider a specific class of random graphs, those created using the famous Erdős–Rényi model \cite{gilbert1959random}.

We now move on to study the entropy of the spanning trees of graphs that are created using the Erdős–Rényi model. In an Erdős–Rényi graph, each edge will be present with a probability of $p$, independent of other edges. As the Erdős–Rényi model does not guarantee connectivity, there will be some graphs that do not even have a spanning tree. Additionally, it is known that there is no closed-form formula to calculate the number of spanning trees of a graph given its number of nodes and edges. Because of these reasons, we will find an upper bound to the entropy of the spanning trees of Erdős–Rényi graphs rather than its actual value.

We consider Erdős–Rényi graphs with $n$ nodes, with parameter $p$. It can easily be shown that the entropy of the graphs created using this model can be calculated using the following equation.
\begin{equation}
\label{ER entropy}
    H_G = {n\choose 2}H(p),
\end{equation}
where
\begin{equation}
    H(p) = -p\log_2p-(1-p)\log_2(1-p).
\end{equation}

Additionally, we assume that once an ER graph is created, one of its spanning trees is chosen uniformly. If the graph does not have any spanning trees, then simply no tree is chosen. The next step is to therefore calculate an upper bound to $H(T|G)$ based on this. Note that as the spanning tree is chosen uniformly after the graph is created, we can write
\begin{equation}
\label{Mihai1}
    H(T|G) = \mathbb{E}\left[\log_2 s(g)\right].
\end{equation}
As logarithm is a concave function, we can apply Jensen's inequality \cite{jensen1906fonctions} to (\ref{Mihai1}) and obtain the following inequality.
\begin{equation}
\label{mihai2}
    H(T|G)\leq \log_2\mathbb{E}\left[s(g)\right].
\end{equation}
Now, we note that based on Cayley's formula, there are $n^{n-2}$ possible trees on $n$ nodes, each of which are a spanning tree of the underlying n-node graph. As the underlying graph is an ER graph, the probability of each of these trees being present in the graph is simply $p^{n-1}$. Therefore, we have
\begin{equation}
\label{mihai3}
    \mathbb{E}\left[s(g)\right] = p^{n-1}n^{n-2}.
\end{equation}
By inserting \ref{mihai3} into \ref{mihai2}, we get the following upper bound on $H(T|G)$.
\begin{equation}
\label{H(T|G)}
\begin{split}
        H(T|G)&\leq \log_2 p^{n-1}n^{n-2}\\
        &=(n-1)\log_2p+(n-2)\log_2n
\end{split}
\end{equation}

We now move on to calculate the term $H(G|T)$ in (\ref{Ht}). Notice that given a spanning tree of an ER graph with $n$ nodes, we will know the status of $n-1$ edges out of the possible $n\choose 2$ edges of the graph. Therefore, given a spanning tree of the graph, the remaining entropy is simply the entropy of the remaining ${n\choose 2}-(n-1)$ edges of an ER graph. Consequently, we can write the following equation.

\begin{equation}
\label{H(G|T)}
    \begin{split}
        H(G|T)&=\sum_{t}p_T(t)H(G|T=t)\\
        &=\sum_{t}p_T(t)\left({n\choose 2}-(n-1)\right)H(p)\\
        &=\left({n\choose 2}-(n-1)\right)H(p)\sum_{t}p_T(t)\\
        &=\left({n\choose 2}-(n-1)\right)H(p),
    \end{split}
\end{equation}
where the sum is taken over all possible spanning trees on the $n$ nodes of the graph, and $p_T$ shows the probability distribution of the trees.

Ultimately, we insert the results from (\ref{ER entropy}), (\ref{H(T|G)}), and (\ref{H(G|T)}) into (\ref{Ht}) to get the total upper bound on the entropy of the spanning trees of the ER model. This will provide us with the following equation after simplification.

\begin{equation}
\label{H(T)}
    H_T\leq (n-1)\left(H(p)+\log_2(np)\right)-\log_2n
\end{equation}

Eq. (\ref{H(T)}) gives us an upper bound on the entropy of trees created using the ER Spanning Tree model. Fig. \ref{ER spanning} illustrates this entropy, and compares it with the entropy of the graph, and the maximum entropy for trees. The maximum entropy is calculated using the fact that a uniform distribution maximises the entropy, and there exist $n^{n-2}$ possible labelled trees on $n$ nodes. The simulation is run for ER graphs with 100 nodes, and the entropy is plotted as a function of the ER parameter $p$. It can be seen that for larger values of $p$, the estimated upper bound for the entropy is larger than the maximum entropy. However, (\ref{H(T)}) is providing us with a tighter upper bound when used for lower values of $p$. The exact boundary for which our upper bound is providing a better bound compared to the maximum possible entropy can be calculated by solving the following equation.
\begin{equation}
    \label{boundary}
    (n-2)\log_2n = (n-1)\left(H(p)+\log_2(np)\right)-\log_2n
\end{equation}
It can easily be checked that the value of $p$ that satisfies (\ref{boundary}) is $p=0.5$. Therefore, our upper bound is working well for values of $p$ less than $0.5$.


\begin{figure}[t!]
    \centering
    \includegraphics[width = 0.8\columnwidth]{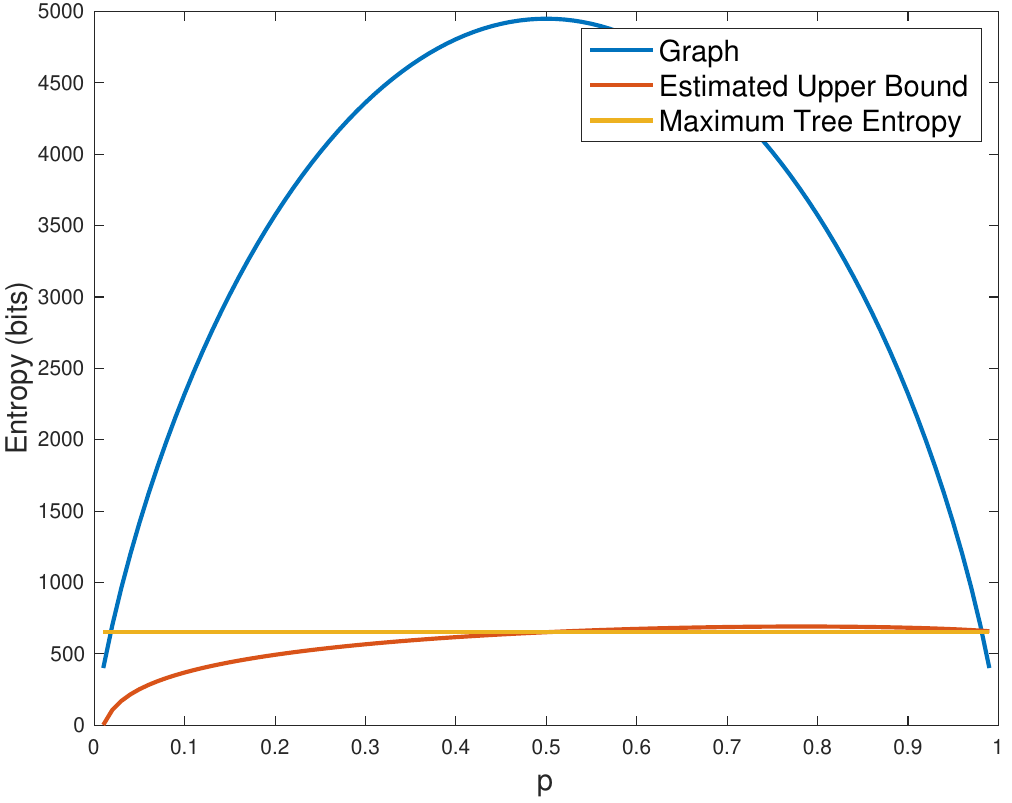}
    \caption{Entropy upper bound for ER spanning trees for graphs with 100 nodes as a function of the ER parameter $p$}
    \label{ER spanning}
\end{figure}

\begin{figure}
    \centering
    \includegraphics[width = 0.8\columnwidth]{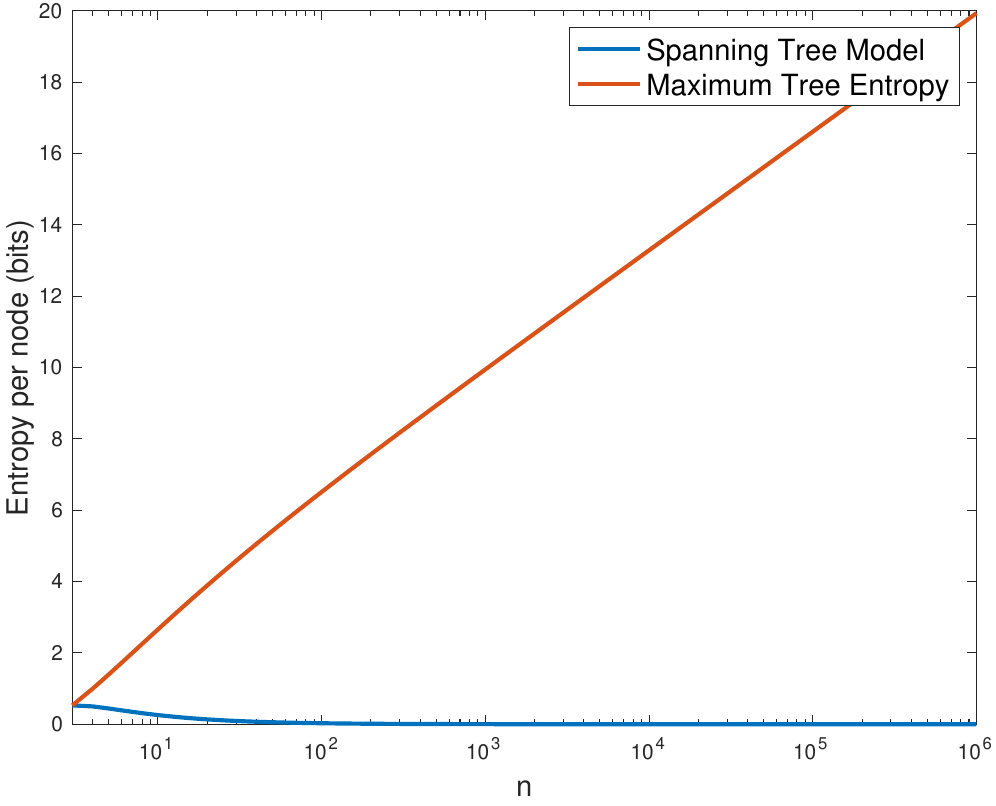}
    \caption{Entropy per node upper bound and maximum entropy per node for ER spanning trees as a function of node number $n$ in the giant component regime threshold}
    \label{connected regime}
\end{figure}

\renewcommand{\thefigure}{4}

Additionally, we performed a simulation for the regime of giant components for ER graphs. According to \cite[Ch.~11.5]{newman2018networks}, an ER graph has a connected component that grows with $n$ (giant component) when $p>1/(n-1)$. This is a regime which is of particular importance in this study, as it is more likely to have a connected graph in this regime. Therefore, we simulated the entropy upper bounds per node (divided by $n$) by setting $p$ to the giant component regime threshold ($1/(n-1)$), and sweeping over $n$. Fig. \ref{connected regime} illustrates the results. It can be seen that as $n$ grows large, our estimate is providing a much tighter upper bound on the entropy of the generated trees. By inserting $p=1/(n-1)$ into (\ref{H(T)}), it can be seen that the entropy upper bound for the giant component regime threshold is $(n-1)\log_2n-(n-2)\log_2(n-2)$.

\section{Universal Compression of Tree Structures}

In this section, we will introduce universal compression algorithms for the models discussed in the previous section. Before doing so, we need to define what we mean by universality and optimality in the context of these methods.

\textbf{Optimality:} In the previous section, we calculated the entropies of different random tree models. It is shown by Shannon \cite{shannon1948mathematical} that the entropy of a random variable provides a lower bound on the average code length that we can use to compress that random variable. Therefore, we are looking for compression algorithms whose average codeword length is close enough to the entropy of the random tree source at hand.

\textbf{Universality:} As seen in the previous section, each of the introduced models for generating random trees has specific parameters. For instance, Simply Generated Trees have their respective children distribution as their parameter. By looking for universal compression algorithms for specific families of random trees, we are essentially looking for compression algorithms that perform optimally regardless of the model parameter. For example, if we develop a universal compression algorithm for Simply Generated Trees, we want it to be able to optimally compress any SGT, regardless of the respective children distribution.

The idea of all existing universal compression algorithms is that as the parameter of the distribution is unknown, the distribution needs to be somehow learned from the data. For instance, the renowned family of Lempel-Ziv \cite{ziv1977universal} compression algorithms, uses dictionaries to store and learn the most common patterns that can happen for the random variable at hand. However, this demands having a sequence of the random variable, so that the most common patterns can be learned. Whereas, in this paper, we are interested in compressing single large trees, rather than a sequence of them. In this case, optimality translated into the average codeword length for all possible trees to be close enough to the entropy of the source. Based on this, we give the following definitions for the optimality of the compression algorithm for different tree source.

\textbf{Sources with a fixed number of nodes:} If we use $E_{L,n}$ to show the expected codeword length for trees with $n$ nodes and $H_n$ to show the entropy of the tree source for trees with $n$ nodes, our goal is for the compressor to satisfy
\begin{equation}
\label{optimality condition}
    \lim_{n\to \infty} E_{L,n} = H_{\text{inf}},
\end{equation}
to ensure that our compression algorithm is asymptotically optimal on single trees when $n$ and consequently the need for compression grow. 

\textbf{Sources with no bound on their number of nodes:} If we use $L(t)$ to show the codeword length of the compressed tree, our goal is for the compressor to satisify the following condition.
\begin{equation}
\label{optimality condition SGT}
    \mathbb{E}[L(t)] = H(T) + C,
\end{equation}
where the expectation is taken over all possible trees in source $T$, and $C$ is a constant with respect to $|T|$.

In order to achieve the conditions in (\ref{optimality condition}) or (\ref{optimality condition SGT}) based on the model, and still have a universal compression algorithm, our main approach will be to decompose each single tree into a sequence of other random variables, and then apply existing universal compression algorithms to those sequences.

Generally speaking, there are numerous universal compression algorithms. The most famous of these algorithms are those designed by Lempel and Ziv in two papers published in 1977 and 1978, which are known as LZ77 \cite{ziv1977universal} and LZ78 \cite{ziv1978compression}, respectively. These algorithms are proven to be asymptotically optimal for stationary and ergodic stochastic processes \cite[Ch.~13]{cover1999elements}. Therefore, if a stochastic process of trees can be considered a stationary and ergodic process, then Lempel-Ziv algorithms can simply be applied to it to get an optimal compression. Another powerful variation of this family of compression algorithms is the Lempel-Ziv-Welch (LZW) algorithm \cite{welch1984technique}, which also satisfies the same conditions. In the following sections, instead of using the term universal compression algorithm, we simply mention LZ. It must be noted by the reader that LZ can simply be replaced by any other universal compression algorithm that is able to optimally compress stationary and ergodic stochastic processes. 

We start by introducing two tree traversal algorithms and their combination, which are shown to be effective for compressing trees from a uniform distribution. We then introduce a universal compression algorithm for Simply Generated Trees, and then move on to Erdős–Rényi random spanning trees. We conclude this section by having a brief discussion on general universal tree compression algorithms.

\subsection{Compressing uniform tree sources}

In this section, we start by proposing two simple, yet effective, tree coding algorithms called Pit-Climbing (PC) and Tunnel-Digging (TD). We will then move on to introduce TreeExplorer, which is a combination of PC and TD. The methods presented in this section are based on our previous work published in \cite{farzaneh2022treeexplorer}.

\subsubsection{Pit-climbing algorithm}
\label{pit-climbing}
In this section, we introduce a novel tree structure coding algorithm that we call \textit{pit-climbing}. We use this term because of the analogy between the proposed method, and a climber that has been trapped in a pit and wants to climb up.

\textbf{Ternary pit-climbing algorithm (TPC): } We start traversing the tree from the leftmost leaf. We log our tree traversal using three symbols: $\uparrow$, $\Uparrow$ and $\downarrow$. Anytime that we are at a leaf, we take the only possible path, which is upwards. If we take an upward path at any point from an edge, we consider that edge and the subtree below it as deleted (or filled-in) from the original tree so that we do not explore it again. Additionally, we log this upward movement in our code. If we have moved to a node that we have never been to before, we log a $\uparrow$ in the code. Otherwise, if we move upwards to a node that we have seen before, we log it with a $\Uparrow$. When we reach a node that is not a leaf, we look at the leaves of the rooted subtree whose root is the node we are currently at. We then take the path downwards that falls into the leftmost leaf of that subtree. We log this entire fall with a single $\downarrow$. We continue exploring the tree and logging the code in the same manner until we reach the root of the tree and there is no other edge to fall into.

We will clarify TPC with the following example.

\begin{exmp}
    Assume that we are given the rooted tree structure of Fig. \ref{tpc}, where the red node shows the root. The starting point of the algorithm is indicated, and the arrows show the path that PC takes. The orange, green, and blue arrows are used to show $\downarrow$, $\uparrow$, and $\Uparrow$, respectively.
    \begin{figure}[H]
        \centering
        \includegraphics[width=\columnwidth]{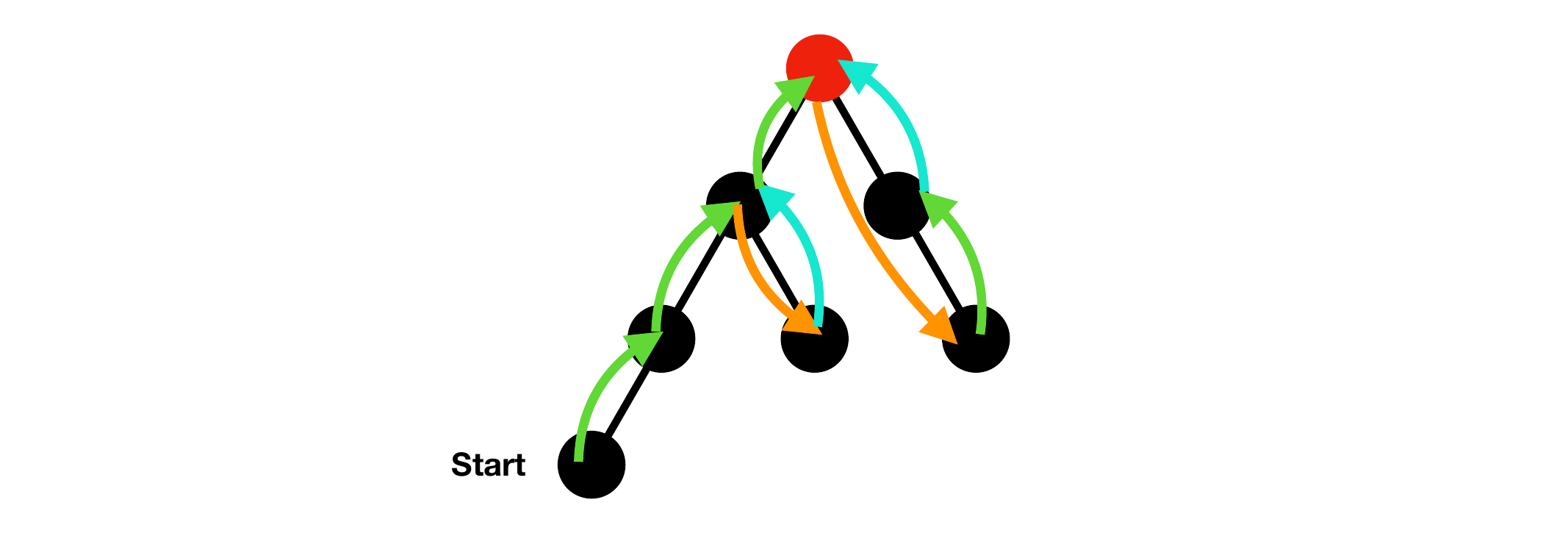}
        \caption{Running TPC on a sample tree}
        \label{tpc}
    \end{figure}
    
\end{exmp}

In source coding, usually a binary code is preferred over a ternary code, as most of our systems for storage and communication are binary-based. To transform TPC codes into binary, we look back at the definition of the symbols. We observe that we can never have consecutive $\downarrow$s. This is because whenever we fall, we fall down to a leaf, so we can never fall twice or more. We make use of this fact, and assign 0 to $\downarrow$, and 00 to $\Uparrow$. We also use 1 to represent $\uparrow$. We call this new binary code for rooted tree structures simply \textit{pit-climbing} (PC). Even though this method of coding does not provide us with an instantaneous code, we claim that PC creates uniquely decodable codes. Theorem \ref{uniquely} proves this statement.

\begin{theorem}
\label{uniquely}
Codes generated by pit-climbing are uniquely decodable.
\end{theorem}
\begin{proof}
We use induction on the depth of the tree. Firstly, the code for a tree with a single node is uniquely decodable ($\emptyset$). Next, assume that we know PC codes for all trees with a depth of $k$ or less are uniquely decodable. For a tree with a depth of $k+1$, we look at the subtrees of the children of the root. Based on the induction, we know that the PC code for all these subtrees are uniquely decodable. The PC code for the original tree is the concatenation of the PC codes of the subtrees, with a connector of $\uparrow\downarrow \equiv 10$ for the first two subtrees, and $\Uparrow\downarrow \equiv 000$ for all other subtrees. In case the root only has one child, the final code will be the code of the subtree rooted at the child, plus an additional $\uparrow \equiv 1$. Therefore, in all of the cases, the PC code of the tree with a depth of $k+1$ can be uniquely decoded.
\end{proof}

Furthermore, we would like to investigate the length of the codewords generated by PC.

\begin{theorem}
\label{PClen}
The PC codeword length for a rooted tree structure with $n$ nodes and $l$ leaves is $n+2l-3$ bits.
\end{theorem}
\begin{proof}
Firstly, notice that each $\downarrow$ in the TPC code corresponds to a leaf, as we always fall into a leaf. Additionally, we fall into every leaf except for the one we start the algorithm from exactly once. Therefore, we have $l-1$ $\downarrow$s in the TPC code, which translates into $l-1$ bits in the PC code. Additionally, we climb up each edge of the tree exactly once. Therefore, the number of $\uparrow$s and $\Uparrow$s in the TPC code is equal to the number of edges, which is $n-1$. However, every $\Uparrow$ will translate into 2 bits in the PC code. The number of $\Uparrow$s is equal to the number of $\downarrow$s, as anytime we fall from a node we will have to climb back up to it at some point. Therefore, the number of $\Uparrow$s is also $l-1$, and we will have $l-1$ additional bits when translating the TPC code into a PC code. Consequently, the total number of bits in the PC code will be $l-1+n-1+l-1=n+2l-3$.
\end{proof}

\subsubsection{Tunnel-digging algorithm}

Based on Theorem \ref{PClen}, the PC code length will increase with the number of leaves. However, the number of rooted trees with $l$ leaves does not necessarily increase with $l$. Hence, there is no justified reason for having longer codes for trees with more leaves. As a result, another algorithm called \textit{tunnel-digging} is developed to tackle this problem. The PC algorithm is based on traversing a tree along its edges, up and down. Whereas, in TD, the aim is to traverse the tree in a horizontal manner. The name \textit{tunnel-digging} comes from seeing the traversal method of this algorithm as digging tunnels between nodes on the same depth.

\textbf{Ternary tunnel-digging algorithm (TTD):} We start with the leftmost child of the root, and start moving right to nodes with the same depth in the order of the nodes. For each node that we encounter, we log a $\leftarrow$ if it is a leaf, and a $\rightarrow$ otherwise. If at any point we have to move between two nodes that are not siblings, we use a $\Rightarrow$ to show the transition (digging a tunnel!). Additionally, if at any point there are no more nodes on the right to move to, we move to the leftmost node on the level below, and we mark this transition again with a $\Rightarrow$. We continue until all the leaves of the tree are logged in the code.

The following example illustrates running TTD on a sample rooted tree structure.

\begin{exmp}
Assume that we are given the rooted tree structure of Fig. \ref{ttd}, where the red node shows the root. The starting point of the algorithm is indicated, and the arrows show the path that TTD takes. The blue, green, and orange arrows are used to show $\leftarrow$, $\rightarrow$, and $\Rightarrow$, respectively.

\begin{figure}[H]
        \centering
        \includegraphics[width=\columnwidth]{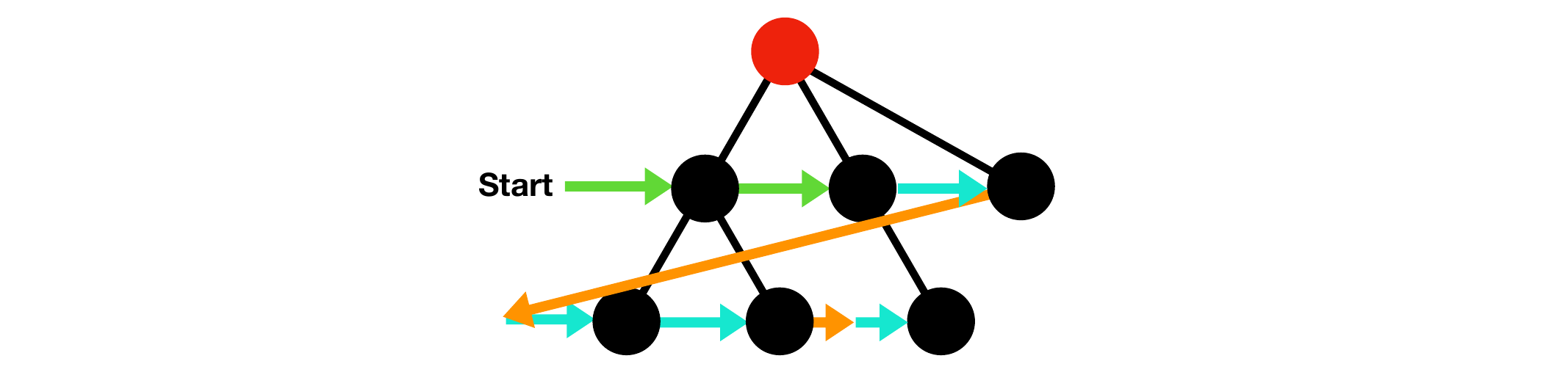}
        \caption{Running TTD on a sample tree}
        \label{ttd}
\end{figure}

\end{exmp}

To transform the TTD code into a binary code which will be called tunnel-digging (TD), we again use the properties of the TTD symbols. We notice that we can never have two consecutive $\Rightarrow$s, as there is always at least one node in between the dug tunnels. Therefore, we use $0$ to represent $\Rightarrow$ and $00$ to represent $\rightarrow$. Additionally, we use $1$ to show $\leftarrow$. Notice that we use a shorter code for leaf nodes, as the number of leaf nodes is expected to be higher than the number of non-leaf nodes when TD is used to code the tree. The proof that the code is uniquely decodable can be done in the same manner as Theorem \ref{uniquely} by replacing $\downarrow$, $\uparrow$, and $\Uparrow$ with $\rightarrow$, $\leftarrow$, and $\Rightarrow$, respectively. The following Theorem calculates the code length of tunnel-digging.

\begin{theorem}
The TD codeword length for a rooted tree structure with $n$ nodes and $l$ leaves is $3n-2l-3$ bits.
\end{theorem}
\begin{proof}
The number of $\rightarrow$s and $\leftarrow$s used in the code is exactly equal to the total number of nodes minus the root. However, we use two bits for each $\rightarrow$, which shows non-leaf nodes. Therefore, the $\rightarrow$s and $\leftarrow$s use $2(n-1)-l$ bits in total. Additionally, for every node that has at least one child (except for the root), we will have a $\Rightarrow$ in the code. This would be equal to $n-1-l$ bits. Thus, we will have $3n-2l-3$ bits in total.
\end{proof}

\subsubsection{TreeExplorer}
\label{hybrid}

To see in which scenarios TD performs better than PC we can write
\begin{equation}
\label{inequa}
    3n-2l-3<n+2l-3 \Rightarrow l>n/2.
\end{equation}
Based on Eq. (\ref{inequa}), PC works better when the number of leaves is less than $n/2$, and TD works better otherwise. They exhibit the same performance when the tree has exactly $n/2$ leaves. Based on this, we propose the following coding technique for rooted tree structures.

\textbf{TreeExplorer:} Firstly, the number of leaves of the rooted tree structure is counted ($l$). If $l< n/2$, the structure is coded with PC. The code is then prefixed with a 0 to specify that it has been coded using PC. Otherwise, the structure is coded using TD, and the code is prefixed with a 1.

It can easily be shown that the codes created using TreeExplorer are uniquely decodable. This is because the first bit of the code uniquely determines the coding method, and we have already proven that both PC and TD are uniquely decodable. The following theorem shows the upper bound for the average code length of TreeExplorer.

\begin{theorem}
For any probability distribution on rooted tree structures with $n$ nodes, the average code length of TreeExplorer is less than $2n-2$.
\end{theorem}
\begin{proof}
Assume that the probability of the number of leaves ($l$) being less than $n/2$ is $q$. If $L$ is the length of the code produced using TreeExplorer, we can write
\begin{align*}
    \mathbb{E}[L] & = q\mathop{\mathbb{E}}_{l<n/2}[n+2l-2]+(1-q)\mathop{\mathbb{E}}_{l\geq n/2}[3n-2l-2]\\
    &< q(2n-2)+(1-q)(2n-2) \\
    & = 2n-2.
\end{align*}
\end{proof}

\subsubsection{Comparison with entropy}

We compare the average code length of TreeExplorer with the entropy of the uniform source, which was calculated in section \ref{entropy_uniform}. The result is plotted in Fig. \ref{entropy_uni}. It can be seen that the performance of TreeExplorer is very close to the entropy of the source, which is the optimal compression limit.

\begin{figure}
    \centering
     \begin{subfigure}[b]{\columnwidth}
         \centering
         \includegraphics[width=0.7\columnwidth]{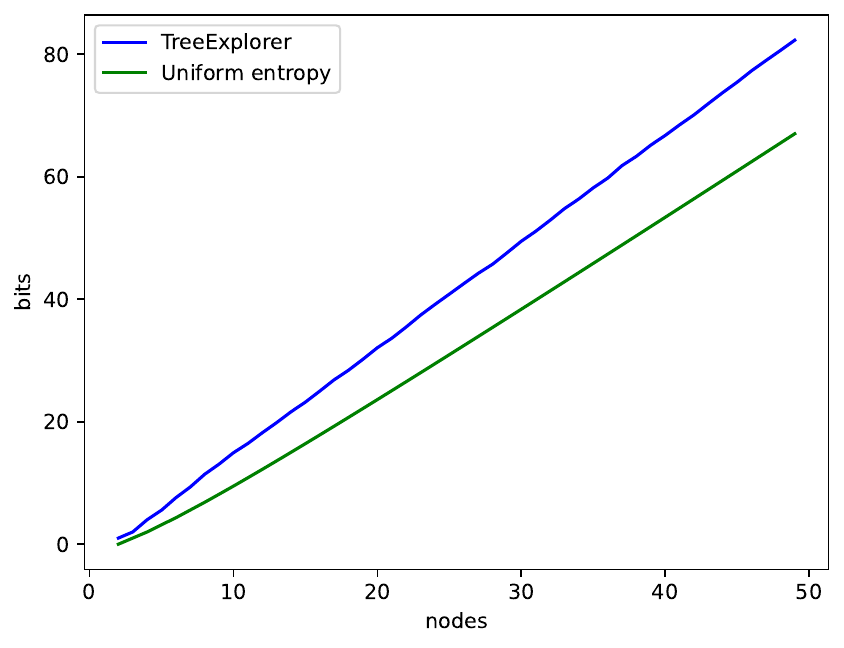}
         \caption{Average length}
         \label{length}
     \end{subfigure}
     \hfill
     \begin{subfigure}[b]{\columnwidth}
         \centering
         \includegraphics[width=0.7\columnwidth]{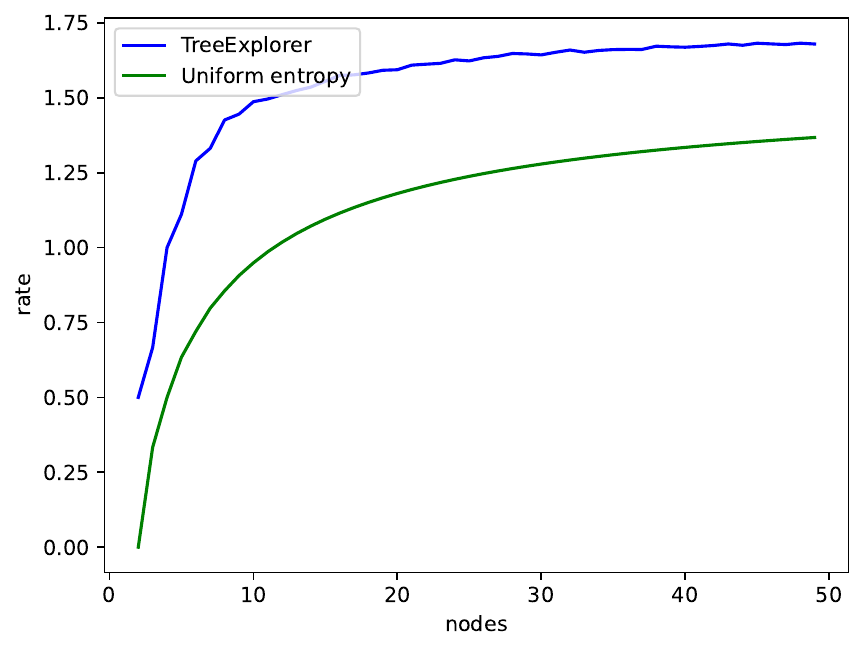}
         \caption{Average rate of change}
         \label{rate}
     \end{subfigure}
     \hfill
     \caption{Comparing the performance of TreeExplorer on a uniform unlabeled unordered rooted tree source with its entropy. Subfigure \ref{length} shows the average bit length, and subfigure \ref{rate} shows its average rate of change.}
    \label{entropy_uni}
\end{figure}

\subsubsection{Comparison with Adjacency list}

The adjacency list representation is one of the most widely used methods for storing trees. This method uses $2n\lceil \log_2 n\rceil$ bits to represent a tree. Fig. \ref{adj list} compares the performance of adjacency list with TreeExplorer. Notice that for coding a labeled tree using TreeExplorer, an additional $n\lceil \log_2 n\rceil$ bits are needed to list the node labels in the order of their appearance.

\begin{figure}
    \centering
    \includegraphics[width = 0.7\columnwidth]{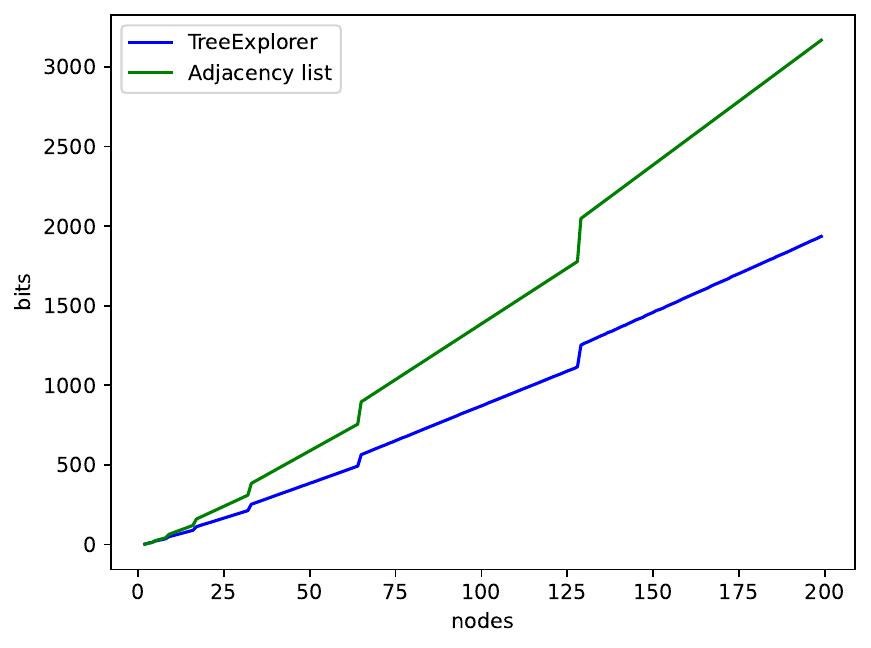}
    \caption{Comparing the performance of TreeExplorer with adjacency list}
    \label{adj list}
\end{figure}

\subsubsection{The Newick format}

The Newick format has been the standard for representing phylogenetic trees since its introduction back in 1986 \cite{cardona2008extended}. In this method, trees are represented using parentheses and commas. This format starts from the root of the tree, and lists the children and subtrees of the root in a nested manner. We will not go into further details on how this method works. However, because of the similarities between this format and TreeExplorer, we compare the performance of these two methods.

Our calculations show an additional $n+1$ bits and $3n-4l+1$ bits in the Newick format compared to when TreeExplorer uses TD and PC, respectively. Fig. \ref{Newick} compares the average codeword lengths of TreeExplorer and the Newick format for rooted trees of up to 50 nodes. It can be seen that the figure also confirms that TreeExplorer provides us with a shorter codeword length.

\begin{figure}[h!]
    \centering
    \includegraphics[width = 0.7\columnwidth]{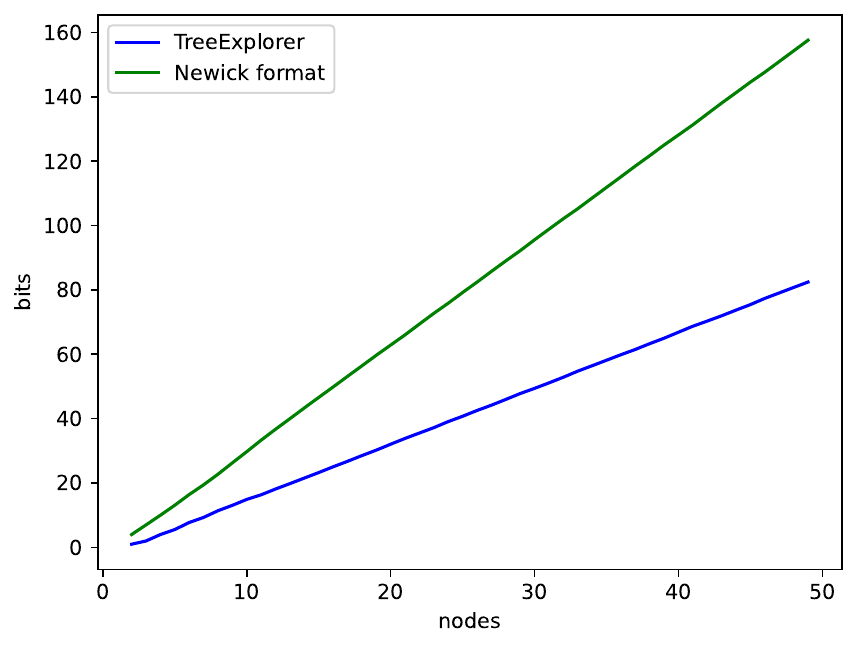}
    \caption{Comparing the performance of TreeExplorer with the Newick format}
    \label{Newick}
\end{figure}

\subsection{Universal compression of SGTs}

In accordance with what we discussed in the previous section, we will be looking for a way to decompose SGTs into a sequence of random variables that can be seen as a stationary and ergodic stochastic process. This will then allow us to simply run a universal compression algorithm on this sequence, and study its performance. We propose the following definition for an SGT sequence, which provides method to create a sequence from a given SGT.

\begin{definition}[SGT sequence]
For a given SGT, its unique SGT sequence is defined as follows. We start from the root, and add the number of its children as the first random variable in the sequence. We then move on to its leftmost child, and add its number of children to the sequence. We continue traversing the tree in a breadth-first manner \cite[Ch.~22.2]{cormen2022introduction}, and keep adding the number of children of each node that we see to the sequence. When we are finished exploring the tree and have added the number of children of all the nodes, the SGT sequence is done.
    
\end{definition}

Fig. \ref{SGT sequence} shows an SGT alongside its SGT sequence. It can easily be shown that an SGT sequence creates a one-to-one mapping between the trees and the sequences, and the tree is fully recoverable given the SGT sequence. We use the notation $f_{SGT}(t)$ to show the SGT sequence of a tree $t$. Additionally, note that the tree traversal can generally be done in any way desirable. In this definition, we have used a Breadth-First Search (BFS). However, any other traversal algorithm such as Depth-First Search \cite[Ch.~22.3]{cormen2022introduction} can also be used. This is of course as long as the same method is used when translating the sequence back into a tree.

\begin{figure}
    \centering
    \includegraphics[width = \columnwidth]{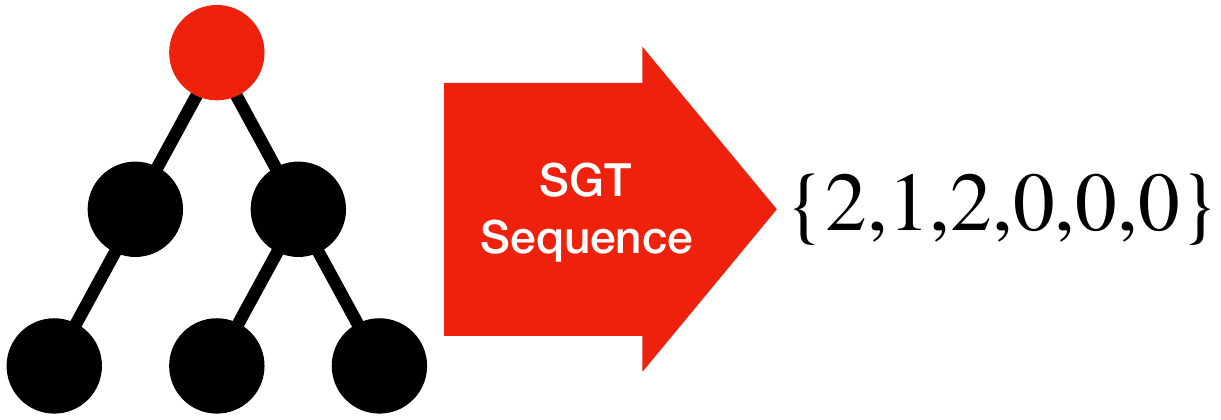}
    \caption{An example SGT and its corresponding SGT sequence}
    \label{SGT sequence}
\end{figure}

Based on the definition of SGTs, the number of children for each node is chosen in an i.i.d manner from the same children distribution. Therefore, the SGT sequence of a tree is a sequence of $n$ i.i.d random variables, with $n$ being the number of children of the tree. As i.i.d sequences are both stationary and ergodic, it can be concluded that LZ can be used to compress SGT sequences.

First, let us consider a sequence of independently generated SGTs $\{t_1,t_2,\ldots, t_n\}$. If we create a new sequence $\{f_{SGT}(t_1), f_{SGT}(t2), \ldots, f_{SGT}(t_n)\}$, it can be seen that the tree sequence can be fully recovered having the SGT sequence. This is because each SGT is uniquely coded using its respective SGT sequence, and the boundaries between consecutive SGTs is always known as the leafs are always marked by having zero children in an SGT. Therefore, the sequence $\{f_{SGT}(t_1), f_{SGT}(t2), \ldots, f_{SGT}(t_n)\}$ is a sequence of i.i.d random variables, and can be compressed optimally using LZ as $n\to \infty$. Therefore, the proposed compression algorithm is both universal and optimal for a sequence of SGTs from the same family.

After having considered sequences of trees, we now move on to study the performance of the proposed compression algorithm on single SGTs. For a single tree $t$, the length of $f_{SGT}(t)$ will also be limited. Therefore, there will be an inevitable redundancy in the LZ code of $f_{SGT}(t)$. We will try to quantify this redundancy in the remainder of this section. To this end, we will use the upper bounds on the redundancy of LZ algorithms, calculated by Savari \cite{savari1997redundancy}. We will focus on the LZW algorithm for now, but the analysis for other members of the LZ family is very similar. Let $R(t)$ show the redundancy per symbol when compressing $f_{SGT}(t)$ using LZW, $\mathscr{R}$ to show the average redundancy per symbol for all the possible trees using the SGT model at hand, and $I(t)$ to show the information in $t$ as per Shannon's definition of information \cite{shannon1948mathematical}. \cite{savariLZW} provides an upper bound on the redundancy of a finite sequence compressed using LZW, which we use to get the following bound on the redundancy for one single tree $t$ with $n$ nodes (and therefore $n$ symbols in its SGT sequence).

\begin{equation}
\label{upper bound single}
    R(t) \leq \frac{I(t)}{n\ln{n}}\log_2\left(\frac{C\log_2e}{H_C} \right)+\mathcal{O}(\frac{1}{\ln{n}})
\end{equation}

To find an upper bound on $\mathscr{R}$, we need to take an average of (\ref{upper bound single}) over all possible trees in the SGT model. To this end, we first note that for $t$ with $n$ nodes, we have $I(t)\leq\log_2n$. If we use $K$ to show $\frac{1}{\ln{2}}\log_2\left(\frac{C\log_2e}{H_C} \right)$, we can write the following upper bound on $R(t)$ based on (\ref{upper bound single}).

\begin{equation}
\label{upper bound single 2}
    R(t) \leq \frac{K}{n}+\mathcal{O}(\frac{1}{\ln{n}})
\end{equation}

We now need to take the average of (\ref{upper bound single 2}) over all possible number of nodes in the SGT model. For this purpose, we need to calculate the probability of a tree having exactly $n$ nodes. This is challenging, and might not be possible. Therefore, we again find an upper bound to it. If we use $p_0$ to show $P_C(c=0)$, and use $N$ as the random variable for the number of nodes in the trees generated by the SGT model, we can write the following equation.

\begin{equation}
\label{pn}
    p(N=n)\leq (1-p_0)^n
\end{equation}

The upper bound in (\ref{pn}) comes from the fact that in order to have $n$ nodes, we can have $n$ nodes to have at least one child. Additionally, we know that for $n=1$, the probability is exactly $p_0$, and $\log_2C$ bits are used to code it in LZW. We also use the lower bound $1-1/x$ to $\ln{x}$ in order to make $\mathcal{O}(1/\ln{x})$ calculable. Having this in mind and combining (\ref{upper bound single 2}) and (\ref{pn}) gives us the following upper bound on $\mathscr{R}$.

\begin{subequations}
\begin{align}
    \mathscr{R}&\leq p_0\log_2C+\sum_{i=2}^{\infty}p(N=i)\left(\frac{K}{i}+\mathcal{O}(\frac{1}{\ln{i}})\right)\\
    &\leq p_0\log_2C+\sum_{i=2}^{\infty}(1-p_0)^i\left(\frac{K}{i}+\frac{1}{1-1/i}\right)\\
    \label{redundancy SGT}
    & = K(p_0-\log_2p_0-1)-\log_2p_0+\frac{1}{p_0}\\
    \nonumber
    &+p_0(1+\log_2p_0+\log_2C)-2
\end{align}
    
\end{subequations}

(\ref{redundancy SGT}) provides us with an upper bound on the average redundancy of using the proposed compression algorithm on a single SGT. It can be seen that this redundancy depends on a number of factors, such as $p_0$, $H_C$, and possible number of children in the SGT model. As this method only compresses single trees, we can not expect the redundancy to be zero, as zero redundancy in universal compression algorithms only happens asymptotically as the number of samples tends to infinity. However, the found upper bound gives an acceptable redundancy for many SGT models. Specifically, it can be seen that the redundancy matches the conditions of optimality defined by (\ref{optimality condition SGT}).

It was seen that the way the number of the children of the nodes is listed in this compression algorithm is based on a DFS traversal on the tree nodes. However, notice that as long as it is known to both the encoder and the decoder how the tree was traversed, it does not matter which tree traversal method is used. This is because the number of children of nodes are i.i.d. random variables, and therefore their order does not affect the optimality of the LZ algorithm. For instance, Breadth-First Search, Pit-Climbing, Tunnel-Digging, or any other tree traversal algorithm that explores all the nodes of the tree could have also been used and we would still get the same results.

\subsection{Universal compression of the Erdős–Rényi Spanning Tree Model}

In this section, we will examine a similar approach to the previous section in order to compress trees generated using the Erdős–Rényi (ER) Spanning Tree Model. Our first aim is to therefore see if we can represent each tree using a sequence of random variables that satisfy stationarity and ergodicity, while using the properties of the way these trees have been generated.

Based on their definition, ER spanning trees are labeled. Additionally, the number of nodes in the graph and the tree must be predetermined ($n$). We propose the following approach for coding ER spanning trees. We divide our coding process into two steps: extracting certain bits from the adjacency matrix of the tree, and then compressing the extracted sequence of bits. We first start by describing the bit extraction process.

\textbf{Bit extraction:} We start by looking at the adjacency matrix of the tree, starting with the row corresponding to the connections of node 1. This row consists of $n-1$ bits $(a_{1,2}, a_{1,3} \ldots, a_{1,n})$, where each bit represents the existence of a connection between node 1 and the other nodes in the tree. We take all these bits during the extraction process. After this, we know all the connections of node 1 in the tree. Let us show the number of these connections with random variables $C_1$. Each pair among these $C_1$ edges removes the possibility of having one other edge in the tree. For instance, if node 1 is connected to both nodes $i$ and $j$, there can not be a connection between $i$ and $j$ in the tree. Therefore, having the connections of node 1 removes the need for including $C_1 \choose 2$ bits in the adjacency matrix of the tree, and we will know exactly which ones. After having described the connections of node 1, we go to the nodes to which node 1 is connected in the order of their labels. We continue by writing down the rows of the adjacency matrix corresponding to these rows, without the connections that have been covered before or whose state is known due to the described edge elimination process. This way, the second node requires less than impossibility of this node being connected to other neighbours of node 1. We can carry on this way and explain the remaining connections of each node, until all the nodes in the tree are covered. Note that bit extraction can be applied to any simple graph, and not just ER graphs. We show the process of bit extraction from a tree $t$ with $f(t)$.

After having extracted certain bits of the adjacency matrix using $f$, we simply feed them into a universal compression algorithm, such as the Lempel-Ziv-Welch algorithm \cite{welch1984technique}. Fig. \ref{ERcompression} summarizes our proposed compression technique.

\begin{figure}
    \centering
    \vspace{0.2in}
    \includegraphics[width = 0.8\columnwidth]{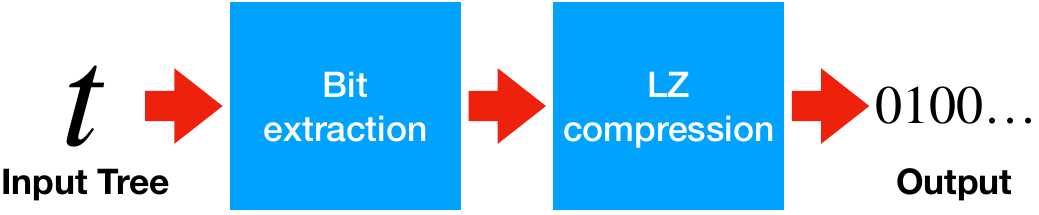}
    \caption{Proposed algorithm for compressing ER Spanning Trees}
    \label{ERcompression}
\end{figure}

We will now show that the redundancy of the proposed compression algorithm tends to zero as the trees grow large. We will perform the calculations for the case where LZ78 is used as the universal compression. However, the result is similar for other compression algorithms in the LZ family. To this end, we state and prove the following theorem.

\begin{theorem}
\label{EROptimal}
The redundancy of the proposed compression algorithm for trees created using the ER spanning tree model tends to zero as the number of nodes of the tree $n$ grow large. The algorithm is also universal in the sense that it does not depend on the value of the ER parameter $p$.
\end{theorem}
\begin{proof}
It is shown in \cite{plotnik1992upper} that for a binary sequence of length $l$, the redundancy of LZ78, which we show with $\mathscr{R}$, satisfies the following inequality.

\begin{equation}
\label{lz78}
    \mathscr{R}\leq \frac{\ln\ln l}{\ln l} + \mathcal{O}\left(\frac{1}{\ln l}\right)
\end{equation}

Therefore, if we use $L(t)$ to show the length of $f(t)$, we can use the following inequality to find the average redundancy of the proposed compression algorithm.

\begin{equation}
\label{lz78_average}
    \mathbb{E}\left[\mathscr{R}\right]\leq \mathbb{E}\left[\frac{\ln\ln L(t)}{\ln L(t)}\right] + \mathbb{E}\left[\mathcal{O}\left(\frac{1}{\ln L(t)}\right)\right]
\end{equation}

As it can be easily shown that $\ln \ln x/\ln x$ is a concave function, we can use Jensen's inequality \cite{jensen1906fonctions} for the first term in (\ref{lz78_average}) and write

\begin{equation}
\label{lz78_Jensen}
    \mathbb{E}\left[\frac{\ln\ln L(t)}{\ln L(t)}\right] \leq \frac{\ln\ln \mathbb{E}[L(t)]}{\ln \mathbb{E}[L(t)]}.
\end{equation}

Based on (\ref{lz78_Jensen}), we need to calculate $\mathbb{E}[L(t)]$. Note that the bit extraction process induces an order on the nodes of tree based on the traversal order. Looking closely, this is simply a Breadth-First Search traversal of the tree, by choosing node 1 as the root. Let us show the number of connections left to be described for the $i$th node in this sequence using $A_i$. Firstly, we know that $A_1 = n-1$. Going to each new node, the need for describing the connections to all the nodes that came before it and all their neighbours is removed. As in each step, we do not have any prior information about the connections to the remaining nodes, each bit can be considered an independent Bernoulli process just like in the original graph. Therefore, for $i>1$ the expected value of unknown connections is $n-1$, minus the $i-1$ node that have come before and their expected number of edges, which is simply $p$ times their number of expected connections. However, we must take into account that this way we are double counting all the prior nodes except for node 1, and this needs a correction term. Based on these reasons, we can write the following recursive equations for the expected values of $A_i$s.

\begin{equation}
\label{sys1}
    \begin{cases}
        \mathbb{E}[A_1] = n-1\\
        \mathbb{E}[A_i] = n-i-p\sum_{j=1}^{i-1}\mathbb{E}[A_j]+(i-2)p,& i>1
    \end{cases}
\end{equation}

Eq. (\ref{sys1}) will result in the following recursive equation for $\mathbb{E}[A_i]$ for $i>1$.

\begin{equation}
\label{sys2}
    \begin{cases}
        \mathbb{E}[A_1] = n-1\\
        \mathbb{E}[A_i] = (1-p)\mathbb{E}[A_{i-1}]-1+p,& i>1
    \end{cases}
\end{equation}

Solving (\ref{sys2}) gives us the following solution.

\begin{equation}
\label{recursive solution}
    \mathbb{E}\left[A_i\right] = \frac{(p-1)^2-\left((n-2)p+1\right)(1-p)^i}{p(p-1)}
\end{equation}
Eq. (\ref{recursive solution}) gives us the following result for the expected number of bits to code.

\begin{dmath}
    \sum_{i=1}^{n}\mathbb{E}[A_i] = \frac{\left(n p^2 - (1 - p)^n - p \left(n (1 - p)^n - 2 (1 - p)^n + 2\right) + 1\right)}{p^2}
    \label{num var}
\end{dmath}

If we use $h(n)$ to show the term calculated in (\ref{num var}), we will only need to code a maximum of $h(n)$ bits from the adjacency matrix of the tree on average. We can write the following equation by inserting (\ref{num var}) into (\ref{lz78_Jensen}).

\begin{equation}
\label{redundancy_final}
    \mathbb{E}[\frac{\ln\ln L(t)}{\ln L(t)}] \leq \frac{\ln\ln h(n)}{\ln h(n)}
\end{equation}

For the $ \mathbb{E}\left[\mathcal{O}\left(\frac{1}{\ln L(t)}\right)\right]$ term in (\ref{lz78_average}), we can simply replace it with a coefficient of $\mathbb{E}[\frac{\ln\ln L(t)}{\ln L(t)}]$ and the inequality will still hold. Therefore, we will have the following upper bound on the average redundancy of the compression algorithm.

\begin{equation}
\label{redundnacy upper bound}
    \mathbb{E}[\mathscr{R}]\leq K\frac{\ln\ln h(n)}{\ln h(n)},
\end{equation}
where $K$ is a constant. It can easily be seen that
\begin{equation}
    \lim_{n\to \infty} h(n) = n.
\end{equation}
Therefore, we can write
\begin{equation}
\label{limitzero}
    \lim_{n\to \infty}\mathbb{E}[\mathscr{R}] = \lim_{n\to \infty} K\frac{\ln\ln h(n)}{\ln h(n)} = 0.
\end{equation}
Eq. \ref{limitzero} shows that the redundancy tends towards zero as $n$ grows large. It can also be seen that it tends towards zero regardless of the value of $p$ and the way that the random spanning trees were chosen. Therefore, it can be said that the proposed method is universal in the sense that it does not depend on the ER parameter or the random tree selection process. 
\end{proof}

An interesting observation is that the proposed compression algorithm can be further generalized by generalizing the bit extraction process. For instance, the process does not necessarily need to start from node 1 as mentioned in the bit extraction procedure. Additionally, it can be seen that other traversal methods such as a DFS traversal would have worked too. The reason for that is the fact that as long as we know the ordering on how the nodes have been traversed, we can reconstruct the respective parts of the adjacency matrix and recover the graph. This gives us the freedom to choose a traversal method that suits the needs of our application better. This can help us in designing a query-preserving tree coding algorithm. In other words, the tree traversal algorithm which determines the ordering of the nodes can be chosen in a way that is able to provide answers to queries of a particular scenario. If we look closely at the proof Theorem \ref{EROptimal}, it can be seen that the only importance of the bit extraction process is that it induces on order on the nodes, and then explores the rows of the adjacency matrix according to that order. Therefore, any traversal on the tree that induces an ordering on all the nodes of the tree can be used in the bit extraction process and we would still get the same results on the optimality of the algorithm. This includes both the traversal algorithm itself, and the choice for the node to start with. Some examples of other algorithms that can be used are any variations of the DFS (Preorder, Inorder, and Postorder) \cite[Ch.~6]{cormen2022introduction}, and the Best-First Search algorithm \cite[p.~48]{pearl1984heuristics}.

\section{Conclusion}

In this paper, we took a novel approach towards the entropy and compression of tree data structures. We started by looking at different random tree sources and analysing their complexity in terms of Shannon entropy. Uniform tree sources and Simply Generated Trees were studies first as existing models of random tree generation. We then moved on to introduce a new random tree generation algorithm that we call the spanning tree model. It was discussed that this model can simulate many of the scenarios that happen in practice. After the entropy of the general model was formulated, we introduced a subcategory of this model whose underlying network is generated using the ER model. This made us able to quantify the entropy of the the source in terms of the model parameters. Ultimately, having the entropy of each of the studied models, we moved on to the compression domain. Universal compression algorithms were introduced for all of the studied models, and it was proven that the redundancy of these algorithms tends to zero as these trees grow large. Future directions of research can include considering other random graph generators for the spanning tree model, and finding a more general tree compression algorithm that goes beyond the ER model.

\section*{Acknowledgment}

This work was supported by EPSRC grant number EP/T02612X/1. For the purpose of open access, the authors has applied a creative commons attribution (CC BY) licence (where permitted by UKRI, ‘open government licence’ or ‘creative commons attribution no-derivatives (CC BY-ND) licence’ may be stated instead) to any author accepted manuscript version arising. We also thank Moogsoft inc. for their support in this research project.

\begin{IEEEbiographynophoto}{Amirmohammad Farzaneh}
(Student Member, IEEE) received his B.Sc. degree at the top of his class in electrical engineering from Sharif University of Technology, Tehran, Iran, in 2020. Since then, he has been a Ph.D. student at the Department of Engineering Science, University of Oxford, Oxford, UK. His research interests include communication systems, information theory, graph theory, and interdisciplinary fields such as computational biology.
\end{IEEEbiographynophoto}

\begin{IEEEbiographynophoto}{Mihai-Alin Badiu}
(Member, IEEE) received the Dipl.-Ing., M.S., and Ph.D. degrees in electrical engineering from the Technical University of Cluj- Napoca, Romania, in 2008, 2010, and 2012, respectively. From 2012 to 2019, he was a Post-Doctoral Researcher and an Assistant Professor with the Department of Electronic Systems, Aalborg University, Denmark. From 2016 to 2018, he held an Individual Post-Doctoral Fellowship from the Danish Council for Independent Research. He is currently a Senior Research Associate with the Department of Engineering Science and a Lecturer in Electrical Engineering at Balliol College, University of Oxford, U.K. His research interests include wireless networks theory, signal processing for communications, and probabilistic modeling and inference.
\end{IEEEbiographynophoto}

\begin{IEEEbiographynophoto}{Justin P. Coon}
(Senior Member, IEEE) received the B.Sc. degree (with distinction) in electrical engineering from Calhoun Honours College, Clemson University, Clemson, SC, USA, in 2000 and the Ph.D. degree in communica- tions from the University of Bristol, Bristol, U.K., in 2005.
He held various research positions with Toshiba Research Europe Ltd. (TREL) from 2004 to 2013, including the position of a Research Manager from 2010 to 2013, during which, he led all research on physical layer communications and network science with TREL. He was a Visiting Fellow with the School of Mathematics, University of Bristol, from 2010 to 2012, and was a Reader with the Department of Electrical and Electronic Engineering from 2012 to 2013. In 2013, he joined the University of Oxford, Oxford, U.K., where he is currently a Professor of Engineering Science and the Emmott Fellow in Engineering Science with Oriel College, Oxford, U.K. He has authored or coauthored more than 200 papers in IEEE and APS journals and conferences and is a named Inventor on more than 30 patents. He was the recipient of the Toshiba’s Distinguished Research Award for his work on 4G systems and three best paper awards. He was the Editor of several IEEE journals and has Chaired or Co-chaired various conferences. He is a Fellow of the Institute for Mathematics and its Applications.
\end{IEEEbiographynophoto}

\end{document}